\documentclass[11pt] {article}
\usepackage{ amsfonts}
\usepackage{ amsthm}
\usepackage{epsfig}
\usepackage{authblk}
\usepackage{cite}
\usepackage{url}

\usepackage{algorithmic,algorithm}

\usepackage{setspace} 

\usepackage[english]{babel}

\usepackage[tight]{subfigure}

\usepackage[cmex10]{amsmath}
\usepackage{amsfonts,amssymb}
\usepackage{bm}

\usepackage{array}

\usepackage{dsfont}

\usepackage{color}

\usepackage[letterpaper,top=.75in,bottom=.75in,left=0.65in,right=0.65in,bindingoffset=0in]{geometry}

\newcommand{\Reals}     {{{\mathrm I\!R}}}  

\newcommand{\define}    {\stackrel{\scriptscriptstyle\triangle}{=}}  
\newcommand{\diag}    {{\mathrm{diag}}}  
\newcommand{\tr}      {{\mathrm{tr}}}    

\newcommand{\Zrb}     {{\uwti 0}}      

\newcommand{\uwti}[1]{{\mathbf #1}}
  \newcommand{\Ab}{{\uwti A}}
\newcommand{\bb}{{\uwti b}}  
\newcommand{\cb}{{\uwti c}}  
  \newcommand{\Db}{{\uwti D}}
\newcommand{\eb}{{\uwti e}}

\newcommand{\hb}{{\uwti h}}  \newcommand{\Hb}{{\uwti H}}
  \newcommand{\Ib}{{\uwti I}}

\newcommand{\rb}{{\uwti r}}  \newcommand{\Rb}{{\uwti R}}
  
\newcommand{\ub}{{\uwti u}}  
\newcommand{\vb}{{\uwti v}}  
  \newcommand{\Wb}{{\uwti W}}
\newcommand{\xb}{{\uwti x}}  
\newcommand{\yb}{{\uwti y}}

          \newcommand{\Sigmab}   {\uwti{\mathnormal\Sigma}}

\newcommand{\Ac} {{\mathcal A}}         
\newcommand{\Bc} {{\mathcal B}}         
\newcommand{\Cc} {{\mathcal C}}         \newcommand{\Ck} {{\bm {\mathcal C}}}

\newcommand{\Fc} {{\mathcal F}}

\newcommand{\Ic} {{\mathcal I}}

\newcommand{\Nc} {{\mathcal N}}         
         
\newcommand{\Pc} {{\mathcal P}}

\newcommand{\Sc} {{\mathcal S}}

\newcommand{\Wc} {{\mathcal W}}         
\newcommand{\Xc} {{\mathcal X}}

\newcommand{\Pulk} {{\underline{{\bm {\mathcal P}}}}}

\newcommand{\Bulk} {{\underline{{\bm {\mathcal B}}}}}
\newcommand{\Tulk} {{\underline{{\bm {\mathcal T}}}}}
\newcommand{\Iulk} {{\underline{{\bm {\mathcal I}}}}}
\newcommand{\Aulc} {{\underline{\mathcal A}}}
\newcommand{\Bulc} {{\underline{\mathcal B}}}

\newcommand{\Uulc} {{\underline{\mathcal U}}}
\newcommand{\Vulc} {{\underline{\mathcal V}}}
\newcommand{\Eulc} {{\underline{\mathcal E}}}
\newcommand{\Fulc} {{\underline{\mathcal F}}}
\newcommand{\Rulc} {{\underline{\mathcal R}}}
\newcommand{\Sulc} {{\underline{\mathcal S}}}

              \newcommand{\eul}  {{\underline e}}

\newcommand{\Iul}  {{\underline I}}

\newcommand{\ba}{\begin{array}}
\newcommand{\ea}{\end{array}}

\newtheorem{definition}{Definition}

\newtheorem{remark}{Remark}
\newtheorem{theorem}{Theorem}
\newtheorem{lemma}{Lemma}
\newtheorem{proposition}{Proposition}

\newtheorem{condition}{Condition}

\begin{document}

\title{Multi-User MIMO Scheduling in the Fourth Generation Cellular Uplink}

\author{Narayan Prasad ~~~Honghai Zhang ~~~Hao Zhu ~~~
Sampath Rangarajan\thanks{Manuscript received May 23, 2012; revised November 6, 2012 and April
24, 2013; accepted June 12, 2013. The associate editor coordinating the review
of this paper and approving it for publication was N. Sagias.
A part of this paper has been accepted for presentation at the IEEE Asilomar
Conference on Signals, Systems, and Computers 2013. Copyright (c) 2013 IEEE. Personal use of this material is permitted. Permission from IEEE must be obtained for all other users, including reprinting/ republishing this material for advertising or promotional purposes, creating new collective works for resale or redistribution to servers or lists, or reuse of any copyrighted components of this work in other works.}
\thanks{ N. Prasad and S. Rangarajan are with NEC Labs America, 4 Independence
Way, Princeton, NJ, USA (e-mail: {prasad, sampath}@nec-labs.com). N.
Prasad is the corresponding author.
H. Zhang is with Google Inc. (e-mail: honghaiz@gmail.com).
H. Zhu is with the Univeristy of Illinois at Urbana Champaign (e-mail:
zhuh@umn.edu).}}

\date{}
\maketitle
\begin{abstract}
We consider Multi-User MIMO (MU-MIMO) scheduling in  the 3GPP LTE-Advanced (3GPP LTE-A) cellular uplink.   The 3GPP LTE-A  uplink  allows for precoded multi-stream (precoded MIMO) transmission from each scheduled user and also allows flexible multi-user (MU) scheduling wherein multiple users can be assigned the same time-frequency resource. However, exploiting these features is made challenging by   certain practical constraints that have been imposed in order to maintain a low signaling overhead. We show that while the scheduling problem in the 3GPP LTE-A cellular uplink is NP-hard, it can be formulated as the maximization of a submodular set function subject to one matroid and multiple knapsack constraints.    We then propose constant-factor polynomial-time   approximation  algorithms and demonstrate their superior performance via simulations.


\end{abstract}
%
%
\newpage

\section{Introduction}

 The 3GPP LTE-A based cellular network \cite{3gpp} together with the IEEE 802.16m based cellular network are the only two cellular networks classified as 4G cellular networks by the international telecommunications union. Some key attributes that a 4G uplink must possess  are the ability to support a peak spectral efficiency of 15 bps/Hz  and a cell average spectral efficiency of 2 bps/Hz,   ultra-low latency and  bandwidths of up to 100MHz. To achieve these ambitious specifications, the 3GPP LTE-A uplink is based on a modified form of the
orthogonal frequency-division multiplexing based multiple-access (OFDMA)  \cite{3gpp}.
 In addition,  it allows precoded multi-stream (precoded MIMO) transmission from each scheduled user as well as  flexible multi-user   scheduling. Notice that while OFDMA itself allows for significant spectral efficiency gains via channel dependent frequency domain scheduling,   multi-user multi-stream communication promises substantially higher degrees of freedom \cite{YuW:degfree}.
Our focus in this paper is on the 3GPP LTE-A uplink (UL) and in particular on MU MIMO scheduling for the  LTE-A UL. Predominantly, almost all of the 4G cellular systems that will be deployed will be based on the 3GPP LTE-A standard \cite{3gpp}.  This standard is an  enhancement of the basic LTE standard which is referred to in the industry as Release 8  and indeed deployments conforming to Release 8 are already underway. The scheduling  in the LTE-A UL is done in the frequency domain where in each scheduling interval the scheduler assigns one or more resource blocks (RBs)   to  each scheduled user. Each RB contains a pre-defined set of consecutive subcarriers and consecutive OFDM symbols and is the minimum allocation unit.

The goal of this work is to design practical uplink MU-MIMO resource
allocation algorithms for the  LTE-A cellular network, where the term resource
refers to RBs as well as precoding matrices. In particular, we consider the design of resource allocation algorithms via weighted sum rate utility maximization that account for   finite user queues (buffers) and finite precoding codebooks. In addition, the designed algorithms comply with all the main practical constraints on the assignment of RBs  and precoders to the scheduled users. We first capture all the key definitions used in this paper in Appendix \ref{app:defns}. Then, we list
our main contributions  in the following:

{\textbf 1)} We first assume that users can employ ideal Gaussian codes and that the base-station (BS) can employ an optimal receiver. We then enforce user rates to lie in a fundamental achievable rate region of the multiple access channel which is a polymatroid and show that the resulting resource allocation problem is NP-hard. We prove that the resource allocation problem can however be formulated as the maximization of a monotonic sub-modular set function subject to one matroid and multiple knapsack constraints, and can be solved using a  recently discovered polynomial time randomized constant-factor approximation  algorithm   \cite{bansal:sparse}. We also adapt a  simpler  deterministic greedy algorithm and show that it  yields a constant-factor approximation  for  scenarios of interest.

        {\textbf 2)} We then consider  scenarios where users employ codes constructed over finite alphabets. In this case the mutual information terms needed to specify  an achievable rate region   do not have closed form expressions. On the other hand the achievable rate region obtained for Gaussian alphabets can be a loose outer bound. Consequently, we obtain a tighter outer bound which is also a polymatroid.  As a result all algorithms developed for Gaussian alphabets can be reused after simple modifications. Finally, we
demonstrate the superior performance of our proposed algorithms via simulations using a realistic channel model.
%

An interesting corollary that follows  from our results is that a popular transmit antenna selection problem in point-to-point MIMO communications  can be posed as a sub-modular maximization problem that is NP-hard but can be approximately solved (with  at-least half optimality) by a simple greedy algorithm.

 \subsection{Related Work}
Resource allocation over  OFDMA networks has been widely studied
 \cite{YuW:dual,OFDMA:Game:Globecom:08,Lee:DLinfo} with  a large fraction of the   problems that have so far been  considered  being single-user (SU) scheduling problems, which attempt to maximize a system utility under  the constraint that scheduled users can only be assigned  non-overlapping  subcarriers. These problems  have been formulated as {\em continuous optimization problems}, and since they are in general non-linear and non-convex,
 many approaches including those
based on  game theory  \cite{OFDMA:Game:Globecom:08} and dual decomposition \cite{YuW:dual}  have been developed.
MU scheduling in the uplink has been considered in \cite{SDMAprasad} which investigates the tradeoff between fairness and efficiency, and from an information theoretic perspective in \cite{capacityPalomar}. In particular, \cite{capacityPalomar} derives a formulation of the capacity region of a discrete memoryless multiple access  channel (MAC) involving only one non-convex constraint and then proposes methods to   compute inner and outer bounds. A MIMO MAC with finite rate feedback is considered in \cite{beamDai}  and a joint user selection, beamforming and quantization strategy is proposed and comprehensively analyzed. 

Recent works have focused on emerging cellular standards and have formulated the respective resource allocation problems as constrained integer programs. A prominent example
 is \cite{Lee:DLinfo} which consider the design of downlink SU-MIMO schedulers for LTE  systems. In this context, we note that downlink frequency domain scheduling in LTE systems using quantized channel quality feedback has been analyzed in  \cite{Donthi:downlink}. On the other hand,  corresponding resource allocation problems for the cellular uplink have been examined in  \cite{multiserver:2009,Yang:ULinfo,prasad:globe11,Lee-UL-2009}. In particular, \cite{multiserver:2009,Yang:ULinfo,Lee-UL-2009} show  that the single-user UL LTE (Release 8) scheduling problem is NP-hard and provide  constant-factor approximation algorithms, whereas \cite{prasad:globe11} considers SU-MIMO LTE-A scheduling. 
The algorithms in \cite{multiserver:2009,Yang:ULinfo,prasad:globe11,Lee-UL-2009}   cannot incorporate MU scheduling  and also cannot incorporate knapsack constraints. MU scheduling for the LTE (Release 8) UL is considered in detail in \cite{prasad:wiopt12}. However, we emphasize that certain additional constraints imposed on LTE (Release 8) UL MU scheduling essentially ensure  that   algorithms optimized for LTE MU-scheduling are unsuitable for LTE-A MU-scheduling whereas algorithms optimized for LTE-A MU-scheduling (as presented in this paper) are not applicable to LTE MU-scheduling. 

\section{MU-MIMO Scheduling in the  LTE-A UL}\label{sec:musetup}
Consider a single-cell uplink with $K$ users and one base-station (BS) which is assumed to have $N_r\geq 1$ receive antennas. Suppose that user $k$ has $N_t\geq 1$ transmit antennas and its power budget is $P_k$. We let  $N$ denote the total number of available resource blocks (RBs). For convenience and without loss of generality, in the following analysis we assume each RB to have unit size.  Then, let $\Hb_{k}^{(n)}$ denote the $N_r\times N_t$ channel matrix seen by the BS from user $k$ on RB $n$, which we assume is known perfectly to the BS. 
 We let $\eul=(u,\cb,\Wb)$ denote a 3-tuple, where $1\leq u\leq K$ denotes a user, $\Wb\in\Wc$ (such that $\tr(\Wb^{\dag}\Wb)=1$) denotes a precoder from a finite codebook $\Wc$ and $\cb\in\Ck$ denotes a valid assignment of RBs chosen from the set $\Ck$ containing all possible valid assignments. In particular, each $\cb$ is an $N-$length vector with binary-valued ($\{0,1\}$) entries and we say an RB $i$ belongs to $\cb$ ($i\in\cb$) if $\cb$ contains a one in its $i^{th}$ position, i.e., $c(i)=1$. Next, we let  $\Eulc=\{\eul=(u,\cb,\Wb):1\leq u\leq K,\cb\in\Ck,\Wb\in\Wc\}$ denote the ground set of all possible such 3-tuples. For any such 3-tuple we adopt the convention that
\begin{eqnarray}
\nonumber \eul=(u,\cb,\Wb)\Rightarrow \cb_{\eul}=\cb,\;\;\Wb_{\eul}=\Wb,\;\;u_{\eul}=u,\;\Hb_{\eul}^{(n)}=\Hb^{(n)}_u\;\forall\;n.
\end{eqnarray}
Suppose now that a subset $\Aulc\subseteq\Eulc$ is selected or scheduled by the base-station. Then on each RB $n$  the received signal vector at the BS can be modeled as the output of a MIMO multiple access channel, as
\begin{equation}\label{eq:propfinmodOr}
\yb^{(n)}=\sum_{\eul\in\Aulc}c_{\eul}(n)\Hb^{(n)}_{\eul}\Wb^{(n)}_{\eul}\xb^{(n)}_{\eul} + \vb^{(n)},
\end{equation}
where
 $\vb^{(n)}\sim\Cc\Nc(\Zrb,\Ib)$ is the additive Gaussian noise and $\xb^{(n)}_{\eul}$ is the input vector corresponding to 3-tuple $\eul$, i.e., the input vector transmitted by user $u_{\eul}$ on RB $n$.

We consider the problem of scheduling users in the frequency domain in a given scheduling interval. 
Let $\alpha_k>0,  \;1\leq k\leq K$ denote the positive weight of the $k^{th}$ user which is an input to the scheduling algorithm and is updated using the output of the scheduling algorithm in every scheduling interval, say according to the proportional fairness rule \cite{Liu-Knightly-2003}. Letting $r_k$ denote the rate assigned to the $k^{th}$ user (in bits per N RBs), we consider the following weighted sum rate utility maximization problem,
\begin{eqnarray}\label{eq:original}
 \max\sum_{1\leq k\leq K}\alpha_{k}r_k,\;\;
\end{eqnarray}
where the maximization is over the assignment of RBs and  precoders to the users
 \textbf{subject to:}

  \textbf{Decodability constraint:} The rates assigned to the scheduled users should be decodable by the base-station receiver. Notice that unlike SU scheduling, MU scheduling allows for multiple users to be assigned the same RB. Thus, the rate that can be achieved for user $k$ need not be only a function of the RBs, precoders and powers assigned to the $k^{th}$ user but can also depend on  those assigned to the other users. 

 \textbf{One precoder and one power level per user:} Each scheduled user can be assigned any one precoding matrix from a finite codebook of such matrices  $\Wc$. In addition, each scheduled user can transmit with only one power level (or power spectral density (PSD))  on all its assigned RBs. This PSD is implicitly determined by the number of RBs assigned to that user, i.e., the user divides its total power equally among all its assigned RBs. The motivation behind these two constraints is that while they significantly decrease the signaling overhead involved in conveying the scheduling decisions to the users, they do not result in significant performance degradation. This is because the uplink channel between each user and the base station is typically highly correlated so that each user's set of preferred spatial directions can be regarded as being approximately frequency non-selective. Consequently, these preferred spatial directions can be reasonably well quantized using a single precoding matrix. Similarly,
    the multi-user diversity effect ensures that each user is scheduled on the set of RBs on which it has relatively good channels. A constant power allocation over such {\em good} channels results in a negligible loss \cite{YuW:Constant}.

  \textbf{At most two chunks per-user:} The set of RBs assigned to each scheduled user should form at-most two mutually non-contiguous chunks, where each chunk is a set of contiguous RBs.  We note here that in the LTE (Release 8) UL each scheduled user is assigned only one  chunk of contiguous RBs \cite{3gpp:rel8}.  Allowing only one chunk of contiguous  RBs to be assigned, together with the DFT spreading operation that each scheduled user must   employ, ensure a low transmit peak-to-average-power ratio (PAPR). In the LTE-A UL each user must employ the DFT spreading but can be assigned up-to two chunks. This relaxation in LTE-A is essentially a compromise between the need to provide more scheduling flexibility and the need to keep PAPR under check \cite{3gpp}.
  A feasible RB allocation and co-scheduling of users in  LTE-A multi-user uplink  is depicted in Fig \ref{fig:lteA}. Notice that each scheduled user is assigned at-most two mutually non-contiguous chunks. Also note that any two scheduled users can partially overlap, i.e., any subset  of the RBs assigned to a user can also be assigned to another user. This is in contrast to the LTE  UL in which any two scheduled users must either not overlap or must completely overlap \cite{3gpp:rel8}.

 \textbf{Finite buffers} We let $Q_k$  denote the size in bits  of the queue (buffer)  associated with the $k^{th}$ user. Thus, the rate $r_k$ assigned to user $k$ cannot exceed $Q_k$.

In addition to the aforementioned constraints the following constraints can also be imposed.

 \textbf{Control channel overhead and interference limit constraints:} Every user that is  scheduled  must be informed about its transmission rate and the set of RBs on which it must transmit, along with the precoder it should employ. This information is sent on the DL control channel of limited capacity which in turn imposes constraints that the set of  scheduled users must respect. These constraints are further discussed in \cite{prasad:wiopt12}. 
 On the other hand, the scheduling decisions that are made should also comply with interference limit constraints which ensure that the  interference caused to other cells does not exceed certain specified margins. 


We will formulate the optimization problem in (\ref{eq:original}) as the {\em maximization of a monotonic submodular set function subject to one matroid and multiple knapsack   constraints}. Towards this end,  we first recall the key definitions from Appendix \ref{app:defns} and then
  enforce that   the non-zero entries in each $\cb\in\Ck$ form at-most two non-contiguous chunks.
In addition, for each 3-tuple $\eul=(u,\cb,\Wb)\in\Eulc$ we let $p_{\eul}$ denote the associated power level (PSD).  This PSD can be computed as $\frac{P_{u}}{{\rm size(\cb)}}$, where ${\rm size(\cb)}$ denotes the number of ones (number of RBs) in $\cb$. Further, let  $\alpha_{\eul},Q_{\eul}$ denote the weight  and buffer (queue) size associated with the 3-tuple $\eul$, respectively and let $r_{\eul}$ denote the rate associated with the 3-tuple $\eul$. We will use the phrase {\em selecting a 3-tuple $\eul$ to imply that
the user $u_{\eul}$ is scheduled to transmit on the RBs indicated in $\cb_{\eul}$ with PSD $p_{\eul}$ and
 precoder $\Wb_{\eul}$}. Thus,  {\em the constraints of one precoder and one power level per user
 along with at most two chunks per-user can be imposed by allowing the scheduler to select any subset of 3-tuples $\Uulc\subseteq\Eulc$ such that $\sum_{\eul\in\Uulc}1\{u_{\eul}=u\}\leq 1$ for each $u\in\{1,\cdots,K\}$, where $1\{.\}$ denotes the indicator function.}
  Accordingly, we define   a family of subsets of $\Eulc$, denoted by $\Iulk$, as
  \begin{eqnarray}\label{eq:Indfamily}
 \Iulk=\left\{\Uulc\subseteq\Eulc: \sum_{\eul\in\Uulc}1\{u_{\eul}=u\}\leq 1,\;\;\forall\;\;1\leq u\leq K \right\}.
 \end{eqnarray}

We recall the model in (\ref{eq:propfinmodOr}) and next consider the decodability constraint after first assuming that each user can employ ideal Gaussian codes (i.e., codes for which the coded modulated symbols can be regarded as i.i.d. Gaussian) and that the BS can employ an optimal receiver. Subsequently, we will consider finite input alphabets. Recall that in DFT-Spread-OFDMA each user linearly transforms its codeword using a DFT matrix in order to reduce the PAPR. Note, however,  that under the assumption of ideal Gaussian codes  and optimal receiver, the DFT spreading operation performed by each user can be ignored.  
  Accordingly, we define a set function $f:2^{\Eulc}\to\Reals_+$ as
\begin{eqnarray}\label{eq:firstSub}
f(\Uulc)=\sum_{n=1}^N\log\left|\Ib + \sum_{\eul\in\Uulc}p_{\eul}c_{\eul}(n)\Hb_{\eul}^{(n)}\Wb_{\eul}(\Hb_{\eul}^{(n)}\Wb_{\eul})^{\dag}\right|,\;\;
\forall\;\Uulc\subseteq \Eulc.
\end{eqnarray}
It can be verified that $f(.)$ defined in (\ref{eq:firstSub}) is a submodular set function, i.e., it satisfies 
\begin{eqnarray*}
f(\Aulc\cup\{\eul\})-f(\Aulc)\geq f(\Bulc\cup\{\eul\})-f(\Bulc),
\end{eqnarray*}
 for all $\Aulc\subseteq\Bulc\subseteq\Eulc$ and $\eul\in\Eulc\setminus\Bulc$.
Further since it is monotonic (i.e., $f(\Aulc)\leq f(\Bulc),\;\forall\;\Aulc\subseteq \Bulc$) and normalized $f(\phi)=0$, where $\phi$ denotes the empty set, we can assert that $f(.)$  is a rank function. Consequently,
for each $\Uulc\subseteq\Eulc$, the region
  \begin{eqnarray}\label{eq:polreg}
  \Pulk(\Uulc,f)=\left\{\rb=[r_{\eul}]_{\eul\in\Uulc}\in\Reals_+^{|\Uulc|}: \sum_{\eul\in\Aulc}r_{\eul}\leq f(\Aulc),\;\forall\;\Aulc\subseteq\Uulc \right\},
 \end{eqnarray}
is a polymatroid \cite{edmonds:poly}. Note that for   each $\Uulc\subseteq\Eulc$, $\Pulk(\Uulc,f)$  is the fundamental achievable rate region of a  multiple access channel so that each rate-tuple $\rb_{\Uulc}=[r_{\eul}]_{\eul\in\Uulc}\in\Pulk(\Uulc,f)$ is achievable \cite{tse:poly}. 
Thus, {\em we can impose decodability constraints by imposing that the assigned rate-tuple satisfy $\rb_{\Uulc} \in\Pulk(\Uulc,f)$ for any selected subset $\Uulc\subseteq\Eulc$.}

Next, in order to impose buffer (queue) constraints, we define a box
 \begin{eqnarray}
  \Bulk(\Uulc)=\{\rb=[r_{\eul}]_{\eul\in\Uulc}\in\Reals_+^{|\Uulc|}: 0\leq r_{\eul}\leq Q_{\eul},\;\forall\;\eul\in\Uulc \},\;\forall\; \Uulc\subseteq\Eulc.
 \end{eqnarray}
{\em Thus, for a (tentative) choice $\Uulc$, we can satisfy both decodability and buffer constraints by assigning only rate-tuples that lie in the region  $ \Pulk(\Uulc,f)\cap \Bulk(\Uulc)$}.
Clearly among all such rate-tuples we are interested in the one that maximizes the weighted sum rate. Hence, without loss of optimality with respect to (\ref{eq:original}), with each $\Uulc\subseteq\Eulc$ we can associate a rate-tuple in $ \Pulk(\Uulc,f)\cap \Bulk(\Uulc)$ that maximizes the weighted sum rate. Accordingly, we define the following set function that determines the reward obtained upon selecting any subset of  $\Eulc$.
We define the set function $h:2^{\Eulc}\to \Reals_+$ as
 \begin{eqnarray}\label{eq:defnh}
  h(\Uulc) = \max_{\rb=[r_{\eul}]_{\eul\in\Uulc}\atop \rb\in \Pulk(\Uulc,f)\cap \Bulk(\Uulc)}  \left\{\sum_{\eul\in\Uulc}\alpha_{\eul}r_{\eul}\right\},\;\;\forall\;\Uulc\subseteq\Eulc.
 \end{eqnarray}


  Leveraging  the arguments made in  \cite{prasad:wiopt12}, 
  we can  represent the control channel overhead constraints as column-sparse knapsack constraints
 such that a subset $\Uulc$ is feasible if and only if
 \begin{eqnarray}
 \Ab_C\xb_{\Uulc}\leq {\bf \bb},
 \end{eqnarray}
 where $\Ab_C\in \{0,1\}^{L\times |\Eulc|}$ is a binary valued matrix for some integer $L\geq 1$ and $ \bb$ is an $L$ length vector, referred to as the control channel budget vector,  whose entries are positive integers. $\xb_{\Uulc}\in\{0,1\}^{|\Eulc|\times 1}$ is a vector with an entry equal to one in each position corresponding to each 3-tuple $\eul\in\Uulc$ and zero elsewhere.
 Notice that the coefficients in $\Ab_C$ are not normalized and hence $\Ab_C$ and $\bb$ together enforce the control channel overhead constraints. Moreover, the total number of non-zero coefficients in any column of $\Ab_C$ is no more than  an integer $\Delta\geq 1$ which denotes the column sparsity level such that $\Delta<<L$.

 Finally, let us consider the interference limit constraints. Suppose that the cell of interest is surrounded by $M$ adjacent cells (or sectors). Let $\eb_m$ be an $N-$length vector of binary valued entries which conveys the RBs such that the total interference caused to the $m^{th}$ base station over all the RBs indicated in $\eb_m$ should be no greater than a specified upper bound. In particular,
  let $\Rb_{u,m}$ be the (wide-band) correlation matrix of the channel seen at the $m^{th}$ base station from the $u^{th}$ user in the cell of interest.\footnote{We assume that the BS in the cell of interest also knows this correlation matrix by exchanging appropriate messages with BS $m$ on the backhaul.} Then the total interference caused to the $m^{th}$ base station over all the RBs indicated in $\eb_m$, upon selecting 3-tuples in any set $\Uulc\subseteq \Eulc$ is equal to
   \begin{eqnarray}\label{eq:interf}
  \sum_{\eul\in\Uulc}\underbrace{p_{\eul}\tr(\Wb_{\eul}^{\dag}\Rb_{u_{\eul},m}\Wb_{\eul})(\cb_{\eul}^T\eb_m)}_{\beta_{\eul,m}}.
 \end{eqnarray}
  Then, we are allowed to select   any set of 3-tuples $\Uulc\subseteq \Eulc$ such that the resulting total interference imposed on the $m^{th}$ base station over all the RBs indicated in $\eb_m$ is
  no greater than a specified upper bound $\gamma_{(m)}$, i.e., such that $\frac{1}{\gamma_{(m)}}\sum_{\eul\in\Uulc}\beta_{\eul,m}\leq 1,\;\forall\; 1\leq m\leq M$.
{\em Thus, all the interference limit constraints can be represented as $M$ generic knapsack constraints
 given by
 \begin{eqnarray}
 \Ab_I\xb_{\Uulc}\leq {\bf 1}_M,
 \end{eqnarray}
 where $\Ab_I\in [0,1]^{M\times |\Eulc|}$ and $ {\bf 1}_M$ is a $M$ length vector of ones.}

Summarizing the aforementioned results, we have formulated (\ref{eq:original})  as the following optimization problem:
 \begin{eqnarray}\label{eq:original2}
\nonumber  \max_{\Uulc\subseteq\Eulc} \{h(\Uulc)\}\;\; {\rm s.t.}\\
 \nonumber\Uulc\in\Iulk;\\
  \Ab_I\xb_{\Uulc}\leq {\bf 1}_M;\;\;\Ab_C\xb_{\Uulc}\leq \bb.
 \end{eqnarray}

 In (\ref{eq:original2}) we regard $M,\Delta$ as constants that are arbitrarily fixed, whereas $L$ can scale polynomially in the cardinality of the ground set $|\Eulc|$. Then, for a given number of users $K$, number of RBs $N$ and the codebook cardinality $|\Wc|$ (which together fix $|\Eulc|$), an instance
(or input) of the problem in  (\ref{eq:original2}) consists of a set of positive user weights $\{\alpha_u\}$ and queue sizes $\{Q_u\}$,  per-user per-RB channel matrices $\{\Hb_u^{(n)}\}:\;1\leq u\leq K,1\leq n\leq N$, a codebook $\Wc$ (of cardinality $|\Wc|$) along with a column sparse matrix $\Ab_C\in \{0,1\}^{L\times |\Eulc|}$, budget vector $\bb$ and any matrix $\Ab_I\in[0,1]^{M\times |\Eulc|}$.
The output is a subset $\hat{\Uulc}\subseteq \Eulc$ along with a rate-tuple $r_{\hat{\Uulc}}$.
Note that  $|\Eulc|$ is $O(K|\Wc|N^{4})$.

We first introduce the following two results that will be invoked later.
\begin{lemma}\label{lem:Lem1}
The family of subsets $\Iulk$ defined in (\ref{eq:Indfamily}) is an independence  family and   $(\Eulc,\Iulk)$ is a partition matroid.
\end{lemma}
\begin{proof} Let $\Eulc_{(k)}$ denote the set of all $\eul\in\Eulc:u_{\eul}=k$ and notice that $\Eulc_{(k)}\cap\Eulc_{(j)}=\phi,\;\forall\;k\neq j$. Then, note that $\Iulk$ can also be defined as $\Aulc\in\Iulk\Leftrightarrow |\Aulc\cap\Eulc_{(k)}|\leq 1\;\forall\;1\leq k\leq K$, which is the definition of a partition matroid (cf. Appendix \ref{app:defns}).
\end{proof}

The proof  of the  following lemma  follows from basic definitions \cite{edmonds:poly} and is skipped for brevity.
 \begin{lemma}\label{lem:Lem2}
 The region $\Pulk(\Uulc,f)\cap \Bulk(\Uulc),\;\forall\;\Uulc\subseteq\Eulc$ is a polymatroid characterized by the rank function $f':2^{\Eulc}\to\Reals_+$ where
 \begin{eqnarray}\label{eq:deffp}
  f'(\Uulc)=\min_{\Rulc\subseteq\Uulc}\left\{f(\Uulc\setminus\Rulc)+\sum_{\eul\in\Rulc}Q_{\eul}\right\},\;\forall\;\Uulc\subseteq\Eulc.
 \end{eqnarray}
 \end{lemma}
We are now ready to offer our main result. Let us assume that computing $h(\Uulc)$ for any $\Uulc\subseteq\Eulc$ incurs a unit cost (or equivalently is given by an oracle in a single query). We will show that even under this assumption the problem in (\ref{eq:original2}) is NP hard. Before proceeding it is useful to recall the definitions given in Appendix \ref{app:defns}.
 \begin{theorem}\label{thm:lteA}
 The optimization problem in (\ref{eq:original2}) is NP hard and is the maximization of a monotonic sub-modular set function subject to one matroid and multiple knapsack constraints.
  \end{theorem}
\proof Proved in Appendix \ref{app:lteA}. \endproof

\begin{theorem}
There is a randomized algorithm whose complexity scales polynomially in $|\Eulc|$ and which yields a  $\frac{e-1}{e^2(M+\Delta+1) + o(M)}$ approximation to (\ref{eq:original2}).
\end{theorem}
\begin{proof} The key observation is that the partition matroid constraint in (\ref{eq:original2}) can be expressed as $K$ knapsack constraints (one for each user). Let $\Ab_P$ denote the resulting $K\times |\Eulc|$ matrix determined by these constraints, whose $k^{th}$ row corresponds to the $k^{th}$ user. Note that this row has ones in each position for which the corresponding 3-tuple $\eul$ satisfies $u_{\eul}=k$ and zeros elsewhere.
Together these $K$ knapsack constraints are column-sparse knapsack constraints wherein in each column  a non-zero entry appears only once.  Thus, the total $K+L+M$ knapsack constraints are column-sparse constraints in which each 3-tuple can appear in at-most $M+\Delta+1$ constraints so that each column can have at-most $M+\Delta+1$ non-zero coefficients.
 With this understanding, we can invoke the randomized algorithm from \cite{bansal:sparse} which is applicable to the maximization of any monotonic submodular set function subject to column-sparse knapsack constraints and obtain the guarantee claimed in the theorem.\end{proof}

Notice that since any monotonic submodular set function is also monotonic and sub-additive, we can infer the following result from  Theorem \ref{thm:lteA}.
\begin{lemma}\label{lem:lemH}
The function $h(.)$ defined in (\ref{eq:defnh}) is sub-additive, i.e.,
\begin{eqnarray}
  h(\Uulc)\leq h(\Uulc_1)+ h(\Uulc_2),\;\forall\;\Uulc_1,\Uulc_2,\Uulc: \Uulc_1\cup\Uulc_2=\Uulc.
   \end{eqnarray}
\end{lemma}

 Practical implementation might demand a  simpler and combinatorial (deterministic) algorithm.  Unfortunately, as remarked in \cite{chandra:submod}, it is difficult to design combinatorial (deterministic) algorithms that can combine both matroid and knapsack constraints. Nevertheless in Algorithm I  we specialize a well known greedy algorithm to our problem of interest (\ref{eq:original2}).
 In this algorithm we maintain a set $\Sulc$. In each iteration of  Algorithm I we add a 3-tuple (from the set of unselected 3-tuples) to $\Sulc$ that yields the largest incremental gain among all feasible 3-tuples that have not yet been selected and where the offered incremental gain is strictly positive. Moreover a 3-tuple is deemed feasible in an iteration if it along with the already selected 3-tuples, satisfies all the constraints in (\ref{eq:original2}). The process continues until either no feasible 3-tuple offers a positive incremental gain or if there are no feasible 3-tuples left.

We now proceed to analyze the performance of Algorithm I and first introduce the following scenario that is of particular interest. We emphasize that this scenario is not required to implement Algorithm I but rather it is introduced since it has a fairly wide applicability and it allows for  a better approximation guarantee. Towards this end, we offer a simple sufficient condition for a knapsack constraint to be matroid constraint.
  \begin{lemma}\label{assump:1}
  The $i^{th}$ knapsack constraint is a matroid constraint if all its strictly positive coefficients are identical,i.e., $1\{A_{i,j}>0\}=1\{A_{i,k}>0\}\Rightarrow A_{i,j}=A_{i,k},\;\forall\;j,k$.
  \end{lemma}
 We  note that
necessary and sufficient conditions for a knapsack constraint (with rational valued coefficients) to be a matroid constraint have been derived in \cite{wolsey:knap} and an efficient algorithm to verify such conditions is given in \cite{barcia:knap}.
Then consider the scenario for which the following two conditions are met.
 \begin{condition}\label{assump1}
   The control channel overhead  constraints are modeled using $L$ knapsack constraints,   where $L$  now represents the number of orthogonal (non-overlapping) control channel regions. Each  user (and hence all its corresponding 3-tuples) is associated with only one of these regions. Further, each constraint corresponds  to a  cardinality constraint   which enforces that no more than a given number of  3-tuples among those associated with the corresponding control region can be scheduled. Notice then that these $L$ control channel overhead constraints are sparse with $\Delta=1$ and since they satisfy Lemma \ref{assump:1} they are   matroid constraints as well. 
   \end{condition}
   We will show in the sequel that when Condition \ref{assump1} is met, the intersection of the $L$ control channel overhead constraints is itself a matroid constraint
  \begin{condition}\label{assump2}
  All the $M$ interference limit knapsack constraints are matroid constraints.
  \end{condition}
  We note that a simplistic modeling of the interference limit constraints can ensure that Condition \ref{assump2} is met.
   For instance, considering the $m^{th}$ interference limit knapsack constraint (corresponding to the adjacent  BS $m$) and recalling (\ref{eq:interf}), each 3-tuple  $\eul\in\Eulc$ can be assigned to one of two sets using an appropriate threshold $\delta_m$: one set comprising those which cause high interference $\{\beta_{\eul,m}>\delta_m\}$ and the other one comprising those which do not. Then a cardinality constraint is imposed only on the set of 3-tuples that cause high interference, i.e., the coefficients (in the $m^{th}$ interference limit knapsack constraint) of all 3-tuples belonging to the first set are set to $1/\gamma_{(m)}$  and the remaining ones are set to zero while the upper bound $\gamma_{(m)}$ is set to be the cardinality bound. Then, it can be seen that all resulting interference limit constraints (upon considering all the $M$ adjacent BSs)  satisfy Lemma \ref{assump:1} and hence are   matroid constraints.

The following result provides the worst-case guarantee offered by Algorithm I.
\begin{theorem}\label{thmG1}
 The complexity of  Algorithm I is $O(K^2N^4|\Wc|)$ and it yields a $\frac{1}{K}$ approximation to (\ref{eq:original2}). Further, if Conditions \ref{assump1} and \ref{assump2} are satisfied  then
Algorithm I yields a constant-factor $\frac{1}{2+M}$ approximation to (\ref{eq:original2}). 
 \end{theorem}

\proof Proved in Appendix \ref{app:thmG1}.

\begin{remark}
Let us reconsider the submodular maximization problem defined in (\ref{eq:Newprob}). This problem in fact  also represents a popular transmit antenna selection problem in point-to-point MIMO communications \cite{sanayei:AS}. Indeed, $K$ can be regarded as the total number of available transmit antennas while $C$ then denotes the number of transmit antennas that have to be selected and a normalization  factor $\sqrt{\frac{\rho}{C}}$, where $\rho$ denotes the SNR, can be absorbed into the matrix $\Hb$. Then, our result in Theorem \ref{thm:lteA}  proves that this transmit antenna selection problem is NP-hard. Next,   the greedy Algorithm I when specialized to this problem reduces to a known incremental successive transmit antenna selection algorithm \cite{sanayei:AS} but for which no approximation guarantees were as yet known. Notice that this problem satisfies Conditions 1 and 2 since the constraint in (\ref{eq:Newprob}) can be accommodated using just one control channel knapsack constraint that has  equal coefficients for all users. Then, invoking the result in Theorem  \ref{thmG1} (with $M=0$) we can infer that the greedy Algorithm I (or equivalently the incremental successive transmit antenna selection algorithm) offers a $1/2$ approximation to the transmit antenna selection problem. An analogous observation for the receive antenna selection problem was made recently and independently in \cite{Txselectvaze}.  In addition, \cite{Txselectvaze}  considers a different version of the  transmit antenna selection problem   in which the number of antennas to be selected, $C$, is not given as an input (but instead is an output) and classifies it as an open problem since it is not equivalent to a submodular maximization problem. We note here that  even for that version, we can obtain an approximation algorithm by sequentially running the greedy Algorithm $K$ times, initialized with inputs $C=1,\cdots,K$ respectively, and picking the overall best among the $K$ outputs. It is readily seen that such an algorithm will also yield a $1/2$ approximation since the output of each run is within $1/2$ of its respective optima.
\end{remark}
\begin{remark}
Recall that we have assumed that the BS employs an ideal receiver, which in practice can be closely approached by iterative Turbo receivers. However, when each user's queue is of infinite size (a.k.a infinitely backlogged case), the assigned rate-tuple   is a corner-point of the polymatroid in (\ref{eq:polreg}) (defined for the selected subset) and thus can be achieved using a simple MMSE-SIC receiver  \cite{VaranasiMK:ODF:Asil97}.
\end{remark}

 Notice that so far we have  assumed that computing $h(\Uulc)$ for any $\Uulc\subseteq\Eulc$ incurs a unit cost. We can indeed show that Algorithm I has polynomial complexity under a stricter notion that computing $f(\Uulc)$ (instead of $h(\Uulc)$) for any $\Uulc\subseteq\Eulc$ incurs a unit cost.\footnote{This assumption results in no loss of generality  since the worst-case cost of computing $f(\Uulc)$ is $O(NK^3)$.} To show this, it suffices to prove that $h(\Uulc)$ can be determined with a complexity
polynomial in $|\Uulc|$.  A key observation towards this end is that for any $\Uulc\subseteq\Eulc$, $f'(\Uulc)$ in (\ref{eq:deffp})  can be computed as
\begin{eqnarray}\label{eq:newfp}
  f'(\Uulc)=\sum_{\eul\in\Uulc}Q_{\eul}+\min_{\Rulc\subseteq\Uulc}\left\{f(\Rulc)-\sum_{\eul\in\Rulc}Q_{\eul}\right\},\;\forall\;\Uulc\subseteq\Eulc.
 \end{eqnarray}
Then, since the function $f(\Rulc)-\sum_{\eul\in\Rulc}Q_{\eul},\;\forall\;\Rulc\subseteq\Eulc$ is a submodular set function, we can  solve the minimization in (\ref{eq:newfp}) using submodular function minimization routines that have a complexity polynomial in $|\Uulc|$ \cite{iwata:submin}.
Thus, from (\ref{eq:newexpH1}) we can conclude that $h(\Uulc)$ can indeed be determined with a complexity
polynomial in $|\Uulc|$.
We now propose simple observations that can considerably speed up Algorithm I.
\begin{itemize}
\item {\em Lazy evaluations.} An important feature that speeds up the greedy algorithm substantially has been discovered and exploited in \cite{Minoux:greedy,krause:greedy}. In particular, due to the submodularity  of the objective function the incremental gain offered by a 3-tuple over any selected subset of 3-tuples not including it decreases monotonically as the selected subset grows larger.
    Thus, at any step in the algorithm, given a set of selected 3-tuples $\Sulc$ and a 3-tuple $\eul\in\Eulc\setminus\Sulc$ for which $h(\Sulc\cup\eul)$ has been evaluated, we do not have to evaluate $h(\Sulc\cup\eul')$ for another 3-tuple $\eul'\in\Eulc\setminus\Sulc$, if we can assert that $h(\Sulc\cup\eul)-h(\Sulc)\geq h(\Sulc'\cup\eul')-h(\Sulc')$ where $\Sulc'\subseteq\Sulc$ denotes the set of selected 3-tuples at a previous step. This results in no loss of optimality with respect to the original greedy algorithm.


\item {\em Exploiting subadditivity.}
Suppose that at any step of the greedy algorithm  we have a set of selected 3-tuples $\Sulc$.
 Further, let $\eul_1=(u,\Wb,\cb_1)$ and $\eul_2=(u,\Wb,\cb_2)$  be two 3-tuples in $\Eulc\setminus\Sulc$ such that $\cb_1$ and $\cb_2$ comprise of only one chunk each and are mutually non-intersecting. Then, letting $\eul'=(u,\Wb,\cb_1+\cb_2)$, we see that
 \begin{eqnarray}
 h(\Sulc\cup\eul')\leq h(\Sulc\cup\eul_1\cup\eul_2)\leq h(\Sulc\cup\eul_1) + h(\Sulc\cup\eul_2)
   \end{eqnarray}
   where the first inequality stems from the fact that $h(\Sulc\cup\eul')$ is monotonically increasing in the transmit PSD of $\eul'$  and the second inequality stems from the monotonicity and subadditivity of $h(.)$. Thus, we have that
   \begin{eqnarray}
 h(\Sulc\cup\eul')\leq 2\max\{ h(\Sulc\cup\eul_1), h(\Sulc\cup\eul_2)\}.
   \end{eqnarray}
 Then if  $\Sulc\cup\eul_1,\Sulc\cup\eul_2$ as well as $\Sulc\cup\eul'$ satisfy all the constraints,
  we can evaluate $h(\Sulc\cup\eul_1), h(\Sulc\cup\eul_2)$ and skip evaluating $h(\Sulc\cup\eul')$. By adopting this procedure over all 3-tuples in   $\Eulc\setminus\Sulc$, we can ensure that the 3-tuple selected will offer at-least $1/2$ the gain yielded by the locally optimal 3-tuple. Then, using a well known result on the greedy algorithm with an approximately optimal selection at each step \cite{nemhaus:analysis} we can conclude that this variation of our greedy algorithm will yield an approximation guarantee of
 $ \frac{1/2}{1/2+M+1}$ when Conditions \ref{assump1} and \ref{assump2} are satisfied.
\end{itemize}
Finally, in order to benchmark the performance of Algorithm I we derive two upper bounds. For convenience, we only consider the case where there are no knapsack constraints so that (\ref{eq:original2}) reduces to the maximization of a monotonic sub-modular set function  subject to one matroid constraint. Then,  we suppose that $\Uulc^{\rm opt}$ and $\hat{\Uulc}$ denote the optimal solution and that returned by Algorithm I. We obtain our first bound by specializing an  upper bound from \cite{Minoux:greedy}   (see also \cite{krause:greedy}) which is applicable to any monotonic sub-modular set function maximization subject to one matroid constraint, as
\begin{eqnarray}\label{eq:UB}
 h(\Uulc^{\rm opt})\leq  h(\hat{\Uulc}) + \sum_{k=1}^K\max_{\eul\in\Eulc_{(k)}\setminus\hat{\Uulc}}(h(\hat{\Uulc}\cup\eul)-h(\hat{\Uulc})),
   \end{eqnarray}
where $\{\Eulc_{(k)}\}$ have been defined in the proof of Lemma \ref{lem:Lem1}.
For our second bound we exhaustively enumerate each one of the $|\Wc|^K$ possible assignments of precoding matrices to users. Then, for each assignment we consider the weighted sum rate maximization over the uplink (\ref{eq:original}) after relaxing the per-user power constraint to one where only a per-user sum power constraint has to be satisfied, i.e., each user can be assigned any power value on any RB as long as it does not exceed its power budget. The latter problem can be efficiently solved via {\em convex optimization} \cite{yuWei:WF,mohseni:jsac}. Finally, we choose the largest weighted sum rate value across all assignments as the upper bound.

\section{Practical Modulation and Coding Schemes}

In the LTE-A uplink a scheduled user can be assigned one out of three modulations ($4,16\;\&\;64$ QAM) and an outer Turbo-code whose coding rate is one out of several available choices. Since the available outer codes are powerful and since the BS can employ near-optimal receivers (such as Turbo SIC) a reasonable choice for the achievable rate region is the following. Let $\Sc_{\eul}$ denote the constellation (with unit average energy and cardinality $S_{\eul}$) associated with 3-tuple $\eul\in\Eulc$. For any subset $\Aulc\subseteq \Eulc$ and any RB $n:1\leq n\leq N$, let $\Ic^{(n)}(\Aulc)$ denote the mutual information evaluated for a point-to-point MIMO channel whose output can be modeled as
\begin{equation}\label{eq:propfinmod}
\yb^{(n)}=\sum_{\eul\in\Aulc}\sqrt{p_{\eul}}c_{\eul}(n)\Hb^{(n)}_{\eul}\Wb^{(n)}_{\eul}\ub^{(n)}_{\eul} + \vb^{(n)},
\end{equation}
 where $\vb^{(n)}\sim\Cc\Nc(\Zrb,\Ib)$ is the additive Gaussian noise and $\ub^{(n)}_{\eul}\in\Sc_{\eul}^{N_t}$ is the input symbol vector corresponding to 3-tuple $\eul$ whose entries are independently and uniformly drawn from $\Sc_{\eul}$ and where $\ub^{(n)}_{\eul},\ub^{(n)}_{\eul'}$ are mutually independent for any $\eul\neq\eul'$.
 Then, for any $\Uulc\subseteq \Eulc$ an achievable rate region is given by
   \begin{eqnarray}\label{eq:region1}
 \left\{\rb=[r_{\eul}]_{\eul\in\Uulc}\in\Reals_+^{|\Uulc|}:\sum_{\eul\in\Aulc}r_{\eul}\leq \sum_{n=1}^N\Ic^{(n)}(\Aulc),\;\forall\;\Aulc\subseteq\Uulc\right\}.
 \end{eqnarray}
 Notice that in deriving (\ref{eq:region1}) we have assumed an ideal BS receiver as well as no DFT spreading by each user, both of which allow for higher achievable rates.\footnote{Neglecting the per-user  DFT spreading expands the rate region since the noise at the BS is assumed to be Gaussian and independent across RBs.}
 Unfortunately, no closed form expressions are available for $\Ic^{(n)}(\Aulc)$
  and the rate region in (\ref{eq:region1}) does not have a useful structure. Clearly the region defined before in (\ref{eq:polreg}) assuming Gaussian inputs is an outer bound which however can be loose. Here we obtain a tighter outer bound that also has a useful structure. We first offer the following result.
  \begin{proposition}\label{propfin}
  For any subset $\Aulc\subseteq \Eulc$  and any $n:1\leq n\leq N$, we have that
  \begin{eqnarray} \label{eq:propfin1}
   \Ic^{(n)}(\Aulc)\leq \underbrace{\min_{\Rulc\subseteq \Aulc}\left\{\log\left|\Ib+\sum_{\eul\in\Aulc\setminus\Rulc}p_{\eul}c_{\eul}(n)\Hb_{\eul}^{(n)}\Wb_{\eul}(\Hb_{\eul}^{(n)}\Wb_{\eul})^{\dag}\right|+  \sum_{\eul\in\Rulc}N_t\log(S_{\eul})\right\}}_{\define g^{(n)}(\Aulc)}
 \end{eqnarray}
 Further the set function $g:2^{\Eulc}\to\Reals_+$ defined as $g(\Aulc)=\sum_{n=1}^Ng^{(n)}(\Aulc),\;\forall\;\Aulc\subseteq\Eulc$, is a rank function.
  \end{proposition}
  \begin{proof}
  Consider any $\Aulc\subseteq \Eulc, n:1\leq n\leq N$ and the model in (\ref{eq:propfinmod}).
     Using the chain rule for mutual information along with the fact that the inputs corresponding to any two distinct 3-tuples of $\Aulc$ are mutually independent, we can upper bound $\Ic^{(n)}(\Aulc)$ as
     \begin{eqnarray*}
     \Ic^{(n)}(\Aulc)\leq   \Ic^{(n)}(\Aulc\setminus\Rulc) + \sum_{\eul\in\Rulc}\Ic^{(n)}(\eul),
      \end{eqnarray*}
      for any $\Rulc\subseteq\Aulc$. Since the cardinality of the input corresponding to 3-tuple $\eul$ is $S_{\eul}^{N_t}$ we have that $\Ic^{(n)}(\eul)\leq N_t\log(S_{\eul})$. Then using the fact   that  for any given input covariance, Gaussian inputs (with the same covariance) maximize the mutual information (over the Gaussian noise channel model in (\ref{eq:propfinmod})), we have that
      \begin{eqnarray*}
       \Ic^{(n)}(\Aulc\setminus\Rulc)\leq \log\left|\Ib+\sum_{\eul\in\Aulc\setminus\Rulc}p_{\eul}c_{\eul}(n)\Hb_{\eul}^{(n)}\Wb_{\eul}(\Hb_{\eul}^{(n)}\Wb_{\eul})^{\dag}\right|. \end{eqnarray*}
       Since these arguments are valid for any subset $\Rulc\subseteq\Aulc$, we can deduce that (\ref{eq:propfin1}) is true.
    The remaining result follows from basic definitions.
      \end{proof}

In this context, we note that the bound in (\ref{eq:propfin1}) is a non-trivial generalization of a bound on the finite alphabet mutual information over a point-to-point fading channel employed in \cite{fabregas:LB} to derive a tight lower bound on the outage probability. However, that bound when applied to our case would only yield $\Ic^{(n)}(\eul)\leq\min\{\log|\Ib+p_{\eul}c_{\eul}(n)\Hb_{\eul}^{(n)}\Wb_{\eul}(\Hb_{\eul}^{(n)}\Wb_{\eul})^{\dag}|,  N_t\log(S_{\eul})\}$ for any $\eul\in\Eulc$.

    Next, we   outer bound the region in (\ref{eq:region1}) as
    \begin{eqnarray}\label{eq:region2}
 \Tulk(\Uulc,g)\define \left\{\rb=[r_{\eul}]_{\eul\in\Uulc}\in\Reals_+^{|\Uulc|}:\sum_{\eul\in\Aulc}r_{\eul}\leq g(\Aulc),\;\forall\;\Aulc\subseteq\Uulc\right\}.
 \end{eqnarray}
  Invoking Proposition \ref{propfin} we use the fact that $g(.)$ is a rank function from which it follows that the region $\Tulk(\Uulc,g)$ is a polymatroid. Then invoking Lemma \ref{lem:Lem2} we can infer the following result.
   \begin{proposition}\label{propfin2}
  For any choice of selected 3-tuples $\Uulc\subseteq\Eulc$, the  rate region
   $\Tulk(\Uulc,g')\define \Tulk(\Uulc,g)\cap \Bulk(\Uulc)$ is a polymatroid
 which is characterized by the rank function
  \begin{eqnarray}\label{eq:gpdef}
 g'(\Aulc)=   \min_{\Rulc\subseteq \Aulc}\left\{g(\Aulc\setminus\Rulc) +\sum_{\eul\in\Rulc}Q_{\eul}\right\}, \;\forall\;\Aulc\subseteq \Uulc.
 \end{eqnarray}
  \end{proposition}

Then, upon by defining
\begin{eqnarray*}
  h'(\Uulc) = \max_{\rb=[r_{\eul}]_{\eul\in\Uulc}\atop\rb\in \Tulk(\Uulc,g') }  \left\{\sum_{\eul\in\Uulc}\alpha_{\eul}r_{\eul}\right\},\;\;\forall\;\Uulc\subseteq\Eulc,
 \end{eqnarray*}
 we consider the optimization problem
 \begin{eqnarray}\label{eq:orig3}
 \nonumber \max_{\Uulc\subseteq\Eulc} \{h'(\Uulc)\}\;\; {\rm s.t.}\\
 \nonumber\Uulc\in\Iulk;\\
  \Ab_I\xb_{\Uulc}\leq {\bf 1}_M;\;\;\Ab_C\xb_{\Uulc}\leq \bb.
 \end{eqnarray}
  As before, it can be shown that the optimization problem in (\ref{eq:orig3}) is the maximization of a monotonic submodular   function subject to one matroid and multiple knapsack constraints. Algorithm I and its associated results are thus applicable.

\section{Simulation Results}

In this section we present our simulation results.
We simulate an uplink  wherein the BS is equipped with four receive antennas and each user has up-to two transmit antennas.  The system has $1024$ sub-carriers out of which  $300$ sub-carriers  divided into 25 RBs (comprising of 12 consecutive sub-carriers each) are available as  data sub-carriers that are used for serving the users.
  We assume 10 active users, all of whom have identical maximum transmit powers and identical path loss factors. We then use the SCM Urban Macro channel model \cite{3gpp}  to  generate the channel between each user and the base-station in an independent identically distributed (i.i.d.) manner. The antenna spacing at the BS is set to be $10\;\lambda$ while that at each user is set to be $1\;\lambda$. In all the results given below  we assume that the BS employs the optimal receiver and each user can employ an unconstrained (Gaussian) input alphabet. Furthermore, unless otherwise mentioned, we assume an infinitely backlogged traffic model wherein each user has an infinite buffer size.  \footnote{We normalize the per-user channels and the noise variance at the BS appropriately and refer to the max transmit power of each user as the (transmit) SNR.} Also, the  per-user weights which are given as inputs to the scheduling algorithm are all set to one so that the objective in (\ref{eq:original}) reduces to the sum rate. We note that since the system considered is homogeneous, fairness among users will also be ensured.

   In Fig. \ref{fig_plot2}, we assume no interference limit or control channel overhead constraints. We first  consider the case where each user is equipped with just one transmit antenna
    and plot the average cell spectral efficiency curve  obtained   when Algorithm I is employed by the BS scheduler. We then consider the case where
    each user is equipped with two transmit antennas and can use an antenna  selection codebook, i.e., $\Wc=\{[1;0],[0;1]\}$ along with the case where an expanded codebook ($\Wc=\{[1;0],[0;1],[1;1]/\sqrt{2},[1;-1]/\sqrt{2},[1;-\sqrt{-1}]/\sqrt{2},[1;\sqrt{-1}]/\sqrt{2}\}$ \cite{3gpp}) can be used for each user. For each curve, we plot a corresponding upper bound using (\ref{eq:UB}).  {\em We caution here that while the upper bound in (\ref{eq:UB}) is very easy to compute, indeed the additional complexity to compute the bound once the solution of Algorithm I is available scales only linearly in the number of users, the bound itself need not be achievable or tight. Its main purpose is to show that the average performance of Algorithm I is significantly superior to its worst-case guarantee, especially over large examples where computing the optimal solution via brute force enumeration is not tractable.} From the figure we observe that in each case, the performance of Algorithm I  is more than $75 \%$ of the  upper bound, which is   superior to the worst case guarantee $1/2$ (obtained by specializing the result in Theorem \ref{thmG1}). Notice that  antenna selection yields a gain of about 1dB over the system with single transmit antenna users while the expanded codebook yields a further gain of about $0.6 \;dB$. However, this additional gain due to the expanded codebook requires an additional power amplifier at each user since simultaneous transmission from both transmit antennas needs to be supported by each user. While antenna selection can be realized with only one power amplifier at each user, in practise it  incurs a switching loss of about $0.4\;dB$. Finally, we note here that the linear increase observed for the spectral efficiency is due to the fact that we have plotted the spectral efficiency versus SNR in dB (or equivalently the logarithm of the absolute SNR).

 In Fig. \ref{fig_plot4} we consider an uplink where each user is equipped with just one transmit antenna
    as well as the case where
    each user is equipped with two transmit antennas and can use an antenna  selection codebook.  We   plot the spectral efficiency obtained upon using Algorithm 1 when each user can be assigned at-most one chunk (enforced by defining the set of feasible RB allocations accordingly) as well as spectral efficiency obtained when each user can be assigned up-to two chunks. From the figure we see that the at-most one chunk restriction does not result in any significant degradation and indeed can be enforced to reduce scheduling complexity as well as to reduce the per-user PAPR. Higher bandwidths (translating to a greater number of available RBs exhibiting greater frequency selectivity) can bring more gain for allowing up-to two chunks per scheduled user.

In Fig. \ref{fig_newUB} we consider the uplink of Fig. \ref{fig_plot4} but where there are seven active users.  We   plot the spectral efficiency obtained upon using Algorithm 1 when each user can be assigned at-most one chunk, along with the corresponding convex optimization based upper bound described in Section \ref{sec:musetup} (referred to in the legend as Imp-UB). It is seen that the performance of Algorithm I is within $5\%$ of this upper bound. While the convex optimization based upper bound is much tighter and reveals the  exceptional performance of Algorithm I, it is computationally demanding to obtain and seems infeasible for larger examples, such as the one in Fig. \ref{fig_plot2} with ten users  and an expanded codebook of  cardinality six.

 In Fig. \ref{fig_plot3} we consider an uplink where there are 15 RBs available for scheduling users and each user is equipped with two transmit antennas and can use an antenna  selection codebook. We impose  a constraint that no more than five users can be scheduled in each scheduling interval. We first plot the spectral efficiency obtained upon using Algorithm 1 with   one control channel overhead constraint to enforce the user limit. In particular, this  constraint has an equal coefficient of $1$ for each user and a budget limit of $5$. Then, we consider two user pre-selection strategies wherein a pool of 5 users is pre-selected in each interval and Algorithm 1 is then used on this pool without any constraints. The intention behind user pre-selection is to reduce scheduling complexity. In the first strategy a greedy rule is employed wherein the reward associated with selecting a user is set equal to the maximum rate that user can offer on any RB and the 5 users with the 5 largest rewards are pre-selected. In the second strategy 5 users are randomly pre-selected. From the figures it is evident that random pre-selection can result in a much degraded performance whereas greedy pre-selection seems a good method to achieve complexity reduction without significant performance degradation.

In the following set of figures we use the 6 path equal gain i.i.d. Rayleigh fading channel model to generate the channel between each user and the base-station.

In Fig. \ref{fig_plotincu} we consider the impact of the number of users ($K$) on the system performance
over an uplink which has  $N=20$ RBs available and wherein each user has one transmit antenna.  We consider two values of  transmit SNRs and first capture the cell spectral efficiency (obtained when Algorithm I is employed by the BS scheduler) as the number of users increases. We then depict the average per-user spectral efficiency.
  Notice first  that the cell spectral efficiency increases only  logarithmically in the number of users since the number of receive antennas is held fixed at four and consequently the per-user spectral efficiency is decreasing in the number of users (i.e., it is $o(K)$). Moreover, we note that in all the cases considered for SNR $18\;dB$, the performance of Algorithm I is more that $75\%$ of the upper bound in (\ref{eq:UB}).

Next, in Figure \ref{fig_plotfb} we assess the impact of finite buffers over the uplink of Fig. \ref{fig_plotincu} but where there are $N=10$ RBs available to service $K=10$ users. In addition, each user can be assigned at-most one chunk of RBs. We assume a  fixed arrival rate per-user which is identical across all users and consider four different values for this arrival rate along with an SNR of $13\;dB$. In each case we plot the cell spectral efficiency obtained when Algorithm I is employed by the BS scheduler, as well as that obtained when a heuristic scheduler is employed. In particular, the heuristic we consider is the one where Algorithm I is first employed assuming infinite buffer sizes. Then, the finite buffer size constraint is imposed separately on each scheduled user.
From the figure we note that at low arrival rates the system is not resource limited in that all users can be simultaneously assigned   rates equal to their respective buffer sizes and any simple scheduling algorithm will suffice. However, at moderate values of arrival rates directly incorporating the buffer sizes in the resource allocation step is quite beneficial. At large values of arrival rates the performance of the heuristic will again approach that of Algorithm I since the buffer size constraints will be increasingly irrelevant.

Finally, in Figure \ref{fig_plotcompnew} we compare the performance of Algorithm I with that of the other algorithms that have been proposed before. In this comparison, we assume that each user has one transmit antenna and can be assigned at-most one chunk and there are $N=20$ RBs available to service the users. We have considered, to the best of our knowledge, all algorithms that   yield  feasible solutions to the problem at hand. In particular, we plot the performance of three algorithms that have been proposed for single-user scheduling over the LTE uplink. These include a greedy heuristic proposed in \cite{Lee-UL-2009}, an approximation algorithm referred to as benefit-doubling (BD) proposed in \cite{multiserver:2009} and another
approximation algorithm based on the local ratio test (LRT) proposed in \cite{Yang:ULinfo}. In addition, we also plot the performance of another approximation algorithm, referred to here as the enhanced local ratio test (ELRT) based algorithm \cite{prasad:wiopt12}, proposed for multi-user scheduling over the LTE uplink where up-to two users can be simultaneously scheduled on an RB provided that any pair of overlapping users are assigned the same set of RBs (a.k.a. the complete overlap constraint). From the figures, we see that Algorithm I yields very significant gains over the previously proposed algorithms. These gains stem from two facts. The first one is that multi-user scheduling over the LTE-A uplink enables substantial gains by allowing multiple users to be co-scheduled on an RB and by relaxing the complete overlap constraint. The second fact is that Algorithm I is near-optimal which allows it to capture almost all of the available gains.

\section{Conclusions}
We considered resource allocation in the 3GPP LTE-A  cellular uplink which allows for MIMO transmission from each scheduled user as well as multi-user  scheduling wherein multiple users can be assigned the same time-frequency resource.  We showed that the resulting resource allocation problem is NP-hard and then   proposed constant-factor polynomial-time   approximation  algorithms.

\appendix

\section{Definitions}\label{app:defns}
We capture some basic known definitions that are invoked in the paper.
\begin{definition}
Given a ground set $\Omega$, we define its power set  (i.e., the set containing all the subsets of  $\Omega$) as $2^{\Omega}$. Then,
a non-negative real valued function defined on the subsets of $\Omega$,
$h:2^{\Omega}\to\Reals_+$ is a {\em monotonic} set function if and only if it satisfies, $0\leq h(\Ac)\leq h(\Bc),\;\forall\;\Ac\subseteq\Bc\subseteq\Omega$. In addition,
the set function is also a
{\em submodular} set function if and only if
\begin{eqnarray}\label{eq:submoddef}
h(\Bc\cup a)-h(\Bc)\leq h(\Ac\cup a)-h(\Ac),\;\forall\;\Ac\subseteq\Bc\subseteq\Omega\;\&\;a\in\Omega\setminus\Bc.
\end{eqnarray}
Furthermore,  the set function is also a
 {\em rank} function function if it is normalized, i.e., $h(\phi)=0$, where $\phi$ denotes the empty set.
 Then, the region defined as $\Pc(\Omega,h)=\left\{\rb=[r_e]_{e\in\Omega}\in\Reals_+^{|\Omega|}: \sum_{e\in\Ac}r_{e}\leq h(\Ac),\;\forall\;\Ac\subseteq\Omega \right\}$ is a {\em polymatroid}.

A {\em knapsack constraint} on the elements of $\Omega$ is a constraint that can be expressed as
 $\sum_{e\in\Omega}a_e\Xc_e\leq b$ for some non-negative scalars $\{a_e\},b\in\Reals_+$ and where $\Xc_e$ is an indicator variable which is one if element $e$ is chosen and zero otherwise. Without loss of generality, we can assume that $a_e\leq b,\;\forall\;e\in\Omega$.
 \end{definition}
\begin{definition}
  $(\Omega,\Iul)$, where $\Iul$ is a collection of some subsets of $\Omega$, is said to be a {\em matroid} if $\Iul$ is an {\em independence family}:
  \begin{itemize}
 \item $\Iul$ is downward closed, i.e., $\Ac\in\Ic\;\&\;\Bc\subseteq\Ac\Rightarrow\Bc\in\Iul$
  \item For any two members $\Fc_1\in\Iul$ and $\Fc_2\in\Iul$ such that $|\Fc_1|<|\Fc_2|$, there exists
   $e\in\Fc_2\setminus\Fc_1$ such that $\Fc_1\cup\{e\}\in\Iul$. This property is referred to as the exchange property.
 \end{itemize}
\end{definition}
\begin{definition}
 $(\Omega,\Iul)$ is said to be
a {\em partition matroid} when   there exists a partition $\Omega=\cup_{i=1}^J\Omega_i$, where $\Omega_i\cap\Omega_j=\phi,\;\forall\;i\neq j$, along with integers $n_i\geq 1\;\forall\;i$ such that
\begin{eqnarray}
 \Bc\subseteq\Omega: |\Bc\cap\Omega_i|\leq n_i\;\forall\;i \Leftrightarrow \Bc\in\Iul
  \end{eqnarray}
\end{definition}
\begin{definition}
An optimization problem is said to be {\em NP-hard} if   any algorithm that returns an optimal solution to the problem at hand given any instance as an input, and whose worst-case complexity (over all instances) scales polynomially in the size of the ground set, can be used to construct such an algorithm for each NP-complete problem. Construction of such algorithms for the latter class of NP-complete problems has been a long standing open problem \cite{Karp_NP72} and indeed the existence of such algorithms is thought to be highly improbable.


A {\em constant factor  approximation algorithm} for a combinatorial optimization problem (in which the objective must be maximized), is an algorithm which returns a feasible solution   given any instance as an input such that the objective value obtained using the returned solution is no less than $\Gamma$ times the optimal objective value for that instance. The factor $\Gamma$ is referred to as the constant-factor and lies in the unit interval  $[0,1]$ and must be independent of the input instance. 
\end{definition}

\section{Proof of Theorem \ref{thm:lteA}}\label{app:lteA}
\begin{proof} We will first show that (\ref{eq:original2}) is the maximization of a monotonic sub-modular set function subject to one matroid and multiple knapsack constraints. Invoking Lemma \ref{lem:Lem1}, it suffices to show that the function $h(.)$ is a monotonic submodular set function. From the definition of $h(.)$ in (\ref{eq:defnh}) it is readily seen that it is monotonic, i.e., $h(\Uulc)\leq h(\Vulc),\;\forall\;\Uulc\subseteq\Vulc\subseteq\Eulc$. There are multiple ways to prove the submodularity of $h(.)$ and we detail one which directly shows that $h(.)$ 
 satisfies  the property in (\ref{eq:submoddef}) for any two subsets $\Uulc\subseteq\Vulc$ in $\Eulc$ and any
 element $\eul\in\Eulc\setminus\Vulc$. Towards this end,
let $o(.,.)$ denote an ordering function such that for any
 subset $\Uulc\subseteq\Eulc$, $o(\Uulc,k)$ is the 3-tuple having the $k^{th}$ largest weight among the 3-tuples in $\Uulc$. Hence we have that $\alpha_{o(\Uulc,1)}\geq \alpha_{o(\Uulc,2)}\geq  \alpha_{o(\Uulc,|\Uulc|)}$. Further, let us adopt the convention that for any subset $\Uulc\subseteq\Eulc$,
  $o(\Uulc,0)=\phi\;\&\;o(\Uulc,k)=\phi,\;\forall\;k\geq |\Uulc|+1\;\&\;\alpha_{\phi}=0$. Defining $J=|\Uulc|$,
    we now invoke Lemma \ref{lem:Lem2} together with the important property that the {\em rate-tuple in each polymatroid that maximizes the weighted sum is determined by the corner point of that polymatroid in which the 3-tuples are arranged in the non-increasing order of their weights \cite{edmonds:poly,tse:poly}}. Thus, we can express $h(.)$ as
    \begin{eqnarray}\label{eq:newexpH1}
   h(\Uulc)=\alpha_{o(\Uulc,1)}f'(o(\Uulc,1))+\sum_{k=2}^{J}\alpha_{o(\Uulc,k)}[f'(\{o(\Uulc,1),\cdots,o(\Uulc,k)\})-f'(\{o(\Uulc,1),\cdots,o(\Uulc,k-1)\})]. \;\;\;\;\;
 \end{eqnarray}
 Let $q$ be the smallest integer in $\{1,\cdots,J\}$ for which $\alpha_{\eul}>\alpha_{o(\Uulc,q)}$ so that $\alpha_{\eul}>\alpha_{o(\Uulc,j)},\;\forall j\geq q$ whereas
   $\alpha_{\eul}\leq \alpha_{o(\Uulc,j)},\;\forall j< q$.
As a result, using (\ref{eq:newexpH1}) we obtain that
 \begin{eqnarray}\label{eq:newdiffuexp}
  \nonumber h(\Uulc\cup\eul)- h(\Uulc)=\alpha_{\eul}[f'(\{o(\Uulc,1),\cdots,o(\Uulc,q-1)\}\cup\eul)-f'(\{o(\Uulc,1),\cdots,o(\Uulc,q-1)\})] +\;\;\;\;\;\\   \nonumber \sum_{k=q}^{J}\alpha_{o(\Uulc,k)}\left[f'(\{o(\Uulc,1),\cdots,o(\Uulc,k)\}\cup\eul)-f'(\{o(\Uulc,1),\cdots,o(\Uulc,k-1)\}\cup\eul) \right.\\ \nonumber\left.   -f'(\{o(\Uulc,1),\cdots,o(\Uulc,k)\})+f'(\{o(\Uulc,1),\cdots,o(\Uulc,k-1)\})\right]\\\nonumber
  =\alpha_{\eul}[f'(\{o(\Uulc,1),\cdots,o(\Uulc,q-1)\}\cup\eul)-f'(\{o(\Uulc,1),\cdots,o(\Uulc,q-1)\})] +\\    \nonumber \sum_{k=q}^{J}\alpha_{o(\Uulc,k)}\left[f'(\{o(\Uulc,1),\cdots,o(\Uulc,k)\}\cup\eul)-f'(\{o(\Uulc,1),\cdots,o(\Uulc,k)\}) \right.\\ \left. -f'(\{o(\Uulc,1),\cdots,o(\Uulc,k-1)\}\cup\eul)  +f'(\{o(\Uulc,1),\cdots,o(\Uulc,k-1)\})\right]
 \end{eqnarray}
which can be re-written as
 \begin{eqnarray}\label{eq:newdiffu}
   \nonumber  h(\Uulc\cup\eul)- h(\Uulc)=(\alpha_{\eul}-\alpha_{o(\Uulc,q)})[f'(\{o(\Uulc,1),\cdots,o(\Uulc,q-1)\}\cup\eul)-f'(\{o(\Uulc,1),\cdots,o(\Uulc,q-1)\})] +\\     \sum_{k=q}^{J}(\alpha_{o(\Uulc,k)}-\alpha_{o(\Uulc,k+1)})[f'(\{o(\Uulc,1),\cdots,o(\Uulc,k)\}\cup\eul)-f'(\{o(\Uulc,1),\cdots,o(\Uulc,k)\})].\;\;\;\;\;
 \end{eqnarray}
Consider now the set $\Vulc$ and suppose that $o(\Vulc,i_k)=o(\Uulc,k),\;\forall\;1\leq k\leq J$, where $1\leq i_1<i_2<\cdots<i_J\leq |\Vulc|$.
Now let $r$ be the smallest integer in $\{1,\cdots,|\Vulc|\}$ for which $\alpha_{\eul}>\alpha_{o(\Vulc,r)}$ and clearly we have $r\leq i_q$.
 Analogous to (\ref{eq:newdiffuexp}), we  express  $h(\Vulc\cup\eul)- h(\Vulc)$ as
  \begin{eqnarray}\label{eq:newdiffvexp}
  \nonumber  h(\Vulc\cup\eul)- h(\Vulc)=\alpha_{\eul}[f'(\{o(\Vulc,1),\cdots,o(\Vulc,r-1)\}\cup\eul)-f'(\{o(\Vulc,1),\cdots,o(\Vulc,r-1)\})] +
   \\ \nonumber   \sum_{k=r}^{|\Vulc|}\alpha_{o(\Vulc,k)}\left[f'(\{o(\Vulc,1),\cdots,o(\Vulc,k)\}\cup\eul)-f'(\{o(\Vulc,1),\cdots,o(\Vulc,k)\}) \right.\\ \left. -f'(\{o(\Vulc,1),\cdots,o(\Vulc,k-1)\}\cup\eul)  +f'(\{o(\Vulc,1),\cdots,o(\Vulc,k-1)\})\right]
 \end{eqnarray}
Due to sub-modularity of $f'(.)$ (cf. property in (\ref{eq:submoddef})) each of the terms corresponding to $k=r,\cdots,|\Vulc|$ in the summation in (\ref{eq:newdiffvexp}) is non-positive. Consequently, we can upper bound
 $h(\Vulc\cup\eul)- h(\Vulc)$ by first dropping the terms corresponding to $k=i_J+1,\cdots,|\Vulc|$ and then reducing the weights of the remaining terms as
  $\alpha_{o(\Vulc,k)}\to \alpha_{o(\Vulc,i_j)}=\alpha_{o(\Uulc,j)},\;\forall\;k:i_{j-1}<k\leq i_j,j\geq q+1$ while $\alpha_{o(\Vulc,k)}\to \alpha_{o(\Vulc,i_q)}=\alpha_{o(\Uulc,q)},\;\forall\;k:r\leq k\leq i_q$.
 Next, we order and parse the remaining $3-$tuples in $\Vulc$ as
 \begin{eqnarray*}
  \left\{\underbrace{o(\Vulc,1),\cdots,o(\Vulc,r-1)}_{\define \Sulc_{q-1}},\underbrace{o(\Vulc,r),\cdots,o(\Vulc,i_q)}_{\define \Sulc_{q}}, \underbrace{o(\Vulc,i_q+1),\cdots,o(\Vulc,i_{q+1})}_{\define \Sulc_{q+1}},\cdots,\underbrace{o(\Vulc,i_{J-1}+1),\cdots,o(\Vulc,i_J)}_{\define \Sulc_{J}}\right\}\;\;\;\;\;\;\;
 \end{eqnarray*}
Combining all terms that have common weights (post the reduction step) we obtain the upper bound to be
\begin{eqnarray}\label{eq:newdiffv}
  \nonumber  h(\Vulc\cup\eul)- h(\Vulc)\leq (\alpha_{\eul}-\alpha_{o(\Uulc,q)})[f'(\Sulc_{q-1}\cup\eul)-f'(\Sulc_{q-1})] +\;\;\;\;\;\;\;\;\;\;\;\;\;\;\;\;\;\;\;\;\;\;\;\;\;\;\;\;\;\;\;\;\;\;\;\;\;\;\;\;
   \\    \sum_{k=q}^{J} (\alpha_{o(\Uulc,k)}-\alpha_{o(\Uulc,k+1)})[f'(\Sulc_{q-1}\cup\cdots\cup\Sulc_k\cup\eul)-f'(\Sulc_{q-1}\cup\cdots\cup\Sulc_k)]
 \end{eqnarray}
Finally, comparing (\ref{eq:newdiffu}) and (\ref{eq:newdiffv}) we note that $(\alpha_{\eul}-\alpha_{o(\Uulc,q)})> 0\;\&\;\{o(\Uulc,1),\cdots,o(\Uulc,q-1)\}\subseteq\Sulc_{q-1}$ and $(\alpha_{o(\Uulc,k)}-\alpha_{o(\Uulc,k+1)})\geq 0\;\&\;\{o(\Uulc,1),\cdots,o(\Uulc,k)\}\subseteq\Sulc_{q-1}\cup\cdots\cup\Sulc_k,\;\forall q\leq k\leq J$. Consequently,
we can invoke the submodularity of $f'(.)$ again to conclude that the upper bound in (\ref{eq:newdiffv}) is less than $h(\Uulc\cup\eul)- h(\Uulc)$ so that $h(\Vulc\cup\eul)- h(\Vulc)\leq h(\Uulc\cup\eul)- h(\Uulc)$, which establishes the submodularity of $h(.)$.

%

We will now show that (\ref{eq:original2}) is an NP hard problem.
We will consider instances of the problem where the number of RBs $N=1$,  all users have identical weights, unit powers, infinite queues and one transmit antenna each and where the codebook $\Wc$ is degenerate, i.e., $\Wc=\{1\}$.  Thus, we have $|\Eulc|=K$. In addition, we assume that the number of receive antennas is equal to the number of users $K$ so that a given  input of user channels forms a $K\times K$ matrix, denoted here by $\Hb=[\hb_1,\cdots,\hb_K]$. Further, we will assume only one knapsack constraint which in particular is a cardinality constraint on the number of users that can be scheduled on the one available RB. We will show that the problem specialized to these instances is also NP-hard so that the original problem is NP-hard.
Note that the matroid constraint now becomes redundant and (\ref{eq:original2}) simplifies to  maximizing the sum rate under a cardinality constraint
\begin{eqnarray}\label{eq:Newprob}
   \max_{\Db=\diag\{d_1,\cdots,d_K\}\atop d_k\in\{0,1\}\;\forall\;k \;\&\;\sum_{k=1}^K d_k\leq C}\log|\Ib+\Hb\Db\Hb^{\dag}|,
   \end{eqnarray}
where $C:1\leq C\leq K$ is the input maximum cardinality.
Now using the determinant equality
\begin{eqnarray}
   \log|\Ib+\Hb\Db\Hb^{\dag}|= \log|\Ib+\Db\Hb^{\dag}\Hb\Db|
   \end{eqnarray}
together with the monotonicity of the objective function, we can re-write (\ref{eq:Newprob}) as
\begin{eqnarray}\label{eq:Newprob2}
   \max_{\Db=\diag\{d_1,\cdots,d_K\}\atop d_k\in\{0,1\}\;\forall\;k \;\&\;\sum_{k=1}^K d_k= C}\log |\Ib+\Db\Hb^{\dag}\Hb\Db|.
   \end{eqnarray}
   Note that (\ref{eq:Newprob2}) is equivalent to determining the $C\times C$ principal sub-matrix of the positive definite matrix $\Ib+\Hb^{\dag}\Hb$ having the maximum determinant. Note that for a given $K$, an instance of the problem in (\ref{eq:Newprob2}) is the matrix $\Hb$ together with $C$. We will prove that (\ref{eq:Newprob2}) is NP-hard via contradiction.
Suppose now that an efficient algorithm (with a complexity polynomial in $K$) exists that can optimally solve (\ref{eq:Newprob2}) for any input $K\times K$ matrix $\Hb$ and any $C:1\leq C\leq K$.
This in turn would imply that there exists an efficient algorithm (with a complexity polynomial in $K$)
that for any input $C:1\leq C\leq K$ and any   $K\times K$ positive definite matrix $\Sigmab$, can determine the  $C\times C$ principal sub-matrix  of $\Sigmab$  having the maximum determinant. 
 Invoking the reduction developed in \cite{ko:algo}, this would then contradict the NP hardness of the problem of determining whether a given input graph has a clique of a given input size.\end{proof}

\section{Proof of Theorem \ref{thmG1}}\label{app:thmG1}
We first consider the complexity of Algorithm I  and note that since the partition matroid constraint needs to be satisfied, there can be at-most $K$ steps in repeat-until loop of the algorithm.
Also, recall that the  the size of the ground set $\Eulc$ is $O(KN^4|\Wc|)$. Then,  at each step we need to compute $h(\Sulc\cup\eul)$ for each $\eul\in\Eulc\setminus\Sulc$ such that $\Sulc\cup\eul$ satisfies all the constraints. Thus, the   worst-case complexity is $O(K^2N^4|\Wc|)$.

Let us now consider the approximation guarantees.
Notice that due to the partition matroid constraint any optimal solution to (\ref{eq:original2}) cannot contain more that $K$ 3-tuples. Then, using the subadditivity of $h(.)$ shown in Lemma \ref{lem:lemH} together with the facts that Algorithm I is monotonic and in its first step selects the 3-tuple of $\Eulc$  having the highest weighted rate, suffice to prove the $\frac{1}{K}$ guarantee. On the other hand, suppose that Conditions \ref{assump1} and \ref{assump2} are satisfied (over all instances). 
Consider the $L$ control channel constraints and  let  $\Eulc_{\ell}$ denote the set of 3-tuples involved in the $\ell^{th}$ control channel constraint so that $\Eulc=\cup_{\ell=1}^L\Eulc_{\ell}$. Recall that $\Eulc_{\ell}\cap \Eulc_{\ell'}=\phi,\;\ell\neq\ell'$ and notice that any set $\Uulc\subseteq\Eulc$ that satisfies these $L$ constraints can be expressed as $\Uulc=\cup_{\ell=1}^L\Uulc_{\ell}$, where $\Uulc_{\ell}\subseteq\Eulc_{\ell}: |\Uulc_{\ell}|\leq b_{\ell},\;1\leq\ell\leq L$, where $b_{\ell}$ is the cardinality bound imposed by the $\ell^{th}$ control channel constraint. Thus the $L$ control channel constraints together are indeed one partition matroid.
More importantly, the intersection of this partition matroid with the one defined in Lemma \ref{lem:Lem1} is also one matroid. To see this, let $\Iulk'$ denote this intersection and recall the definitions given in Appendix \ref{app:defns}. It can readily be seen that $\Iulk'$ is downward closed. Then, we need to show that the exchange property holds. Consider any
 $\Fulc_1,\Fulc_2$ in $\Iulk'$ such that $|\Fulc_1|< |\Fulc_2|$.  Clearly, the users corresponding to all 3-tuples in $\Fulc_1$ must all be distinct since $\Fulc_1\in\Iulk$. In addition, each 3-tuple of $\Fulc_1$ can   have a non-zero coefficient in only one control channel constraint. Similarly for $\Fulc_2$. Then consider any 3-tuple $\eul\in\Fulc_2\setminus\Fulc_1$ such that no 3-tuple in $\Fulc_1$ contains the user $u_{\eul}$. Notice that there must exist at-least one such 3-tuple. Clearly, for such a 3-tuple $\Fulc_1\cup\{\eul\}\in\Iulk$. Consequently, $\Fulc_1\cup\{\eul\}\notin\Iulk'$ only if
  a control channel constraint  is violated. Without loss of generality, suppose this constraint is the first control channel constraint. Then, since all non-zero coefficients in any control channel constraint are identical, we can deduce that there exists a 3-tuple $\eul'\in\Fulc_1$   such that $\eul'\in\Eulc_1$ but the user
   $u_{\eul'}$ is not contained in any 3-tuple of $\Fulc_2$. This observation together with the fact that $|\Fulc_1|<|\Fulc_2|$ allows us to conclude that there exists an $\eul\in\Fulc_2\setminus\Fulc_1$ such that   $\Fulc_1\cup\{\eul\}\in\Iulk'$, which then yields the desired result.

Finally, combining this matroid with the  other $M$ (interference limit) matroid constraints, we see that the feasible subsets belong to the intersection of $M+1$ matroids and hence form a $p-$system where $p=M+1$. Then invoking the guarantee offered by the greedy algorithm on a $p-$system \cite{nemhaus:analysis,nemhaus:algo}, proves the second part.\endproof

\bibliography{SCFDMA_bibU}
\bibliographystyle{ieeetr}
\bibliographystyle{IEEEtran}




\begin{figure}
\centering
\includegraphics[width=0.9\linewidth]
{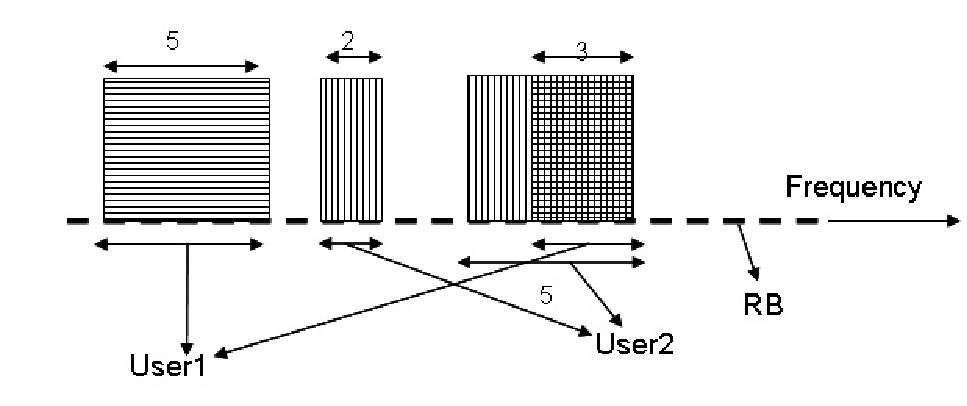} \caption{A Feasible RB Allocation in the LTE-A UL: The assignment of RBs to each user is represented by a shaded region.} \label{fig:lteA}
\end{figure}

\algsetup{indent=1em}
\begin{table}
\caption{{\bf Algorithm I: Greedy Algorithm  for LTE-A UL MU-MIMO}}\label{algo:glteA}
\begin{algorithmic}[1]
\STATE Initialize $\Sulc=\phi$
\STATE \textbf{Repeat}
\STATE Determine
\begin{eqnarray}
  \hat{\eul}=\arg\max_{\eul\in\Eulc\setminus\Sulc\atop \Sulc\cup\eul\in\Iulk;\Ab_I\xb_{\Sulc\cup\eul}\leq {\bf{1}}_M;\Ab_C\xb_{\Sulc\cup\eul}\leq {\bb}}\left\{h(\Sulc\cup\eul)\right\}
 \end{eqnarray}
  and set $\hat{v}=h(\Sulc\cup\hat{\eul})-h(\Sulc)$.
 \STATE \textbf{If} $\hat{v}>0$ \textbf{Then}
 \STATE $\Sulc\leftarrow\Sulc\cup\eul$
 \STATE \textbf{End If}
\STATE   \textbf{Until} $\hat{v}\leq 0$ or $\hat{\eul}=\phi$
\STATE Output    $\Sulc$.
\end{algorithmic}
\end{table}

\begin{table}
\begin{footnotesize}
\begin{center}
\hspace{-1cm}\caption[]{Symbol Definitions}\label{tab:defnsym}
\begin{tabular}{|c|p{4cm}|p{4cm}|p{4cm}|}\hline
   $K$ & Number of users &  $N$ & Number of RBs \\ \hline  $N_t$ & Number of TX antennas at each user &
   $N_r$ & Number of RX antennas at BS \\ \hline  $L$ & Number of column sparse knapsack constraints that model the control channel overhead constraints  & $\Ab_C\in\{0,1\}^{L\times|\Eulc|}$ & Matrix containing the coefficients of the column-sparse knapsack constraints   \\ \hline
   $\Delta$ & column sparsity level in $\Ab_C$ & $\bb$ & $L$ length control channel budget vector       \\  \hline $M$ & number of generic knapsack constraints &     $\Ab_I\in [0,1]^{M\times|\Eulc|}$ & Matrix containing the normalized coefficients of the generic knapsack constraints   \\\hline
 $\alpha_u$ & Weight of  user $u$   & $r_u$ & rate (bits/frame) assigned to user $u$ \\ \hline   $P_u$ & Power budget of user $u$ & $Q_u$ & Buffer size of user $u$ \\\hline
     $\cb$ & $N$-length vector representing a valid RB assignment containing at-most two chunks &   $\Wb$  & Precoder matrix having unit Frobenius norm \\ \hline  $\Wc$ & Finite codebook of all precoder matrices &
     $\Ck$ & Set of all valid RB assignments \\ \hline  $\eul=(u, \cb,\Wb)$ & 3-tuple denoting allocation of RB assignment $\cb$ and precoder $\Wb$ to user $u$ & $\Eulc$ & Ground set containing all possible 3-tuples \\ \hline  $u_{\eul}$ & user in 3-tuple $\eul$ &  $\cb_{\eul}$ & RB assignment in 3-tuple $\eul$ \\ \hline  $\Wb_{\eul}$ & precoder in 3-tuple $\eul$ &
     $\Iulk$ & Collection of valid subsets of $\Eulc$ \\ \hline  $\Hb_{\eul}^{(n)}$ & Channel matrix seen from user $u_{\eul}$ on RB $n$  & $\Bulk(\Uulc)$ & Region defined by buffer sizes of 3-tuples in $\Uulc$         \\  \hline
          $f(.),f'(.),g(.),g'(.)$ & four different rank functions & $\Pulk(\Uulc,f),\Pulk(\Uulc,f')$, $\Tulk(\Uulc,g),\Tulk(\Uulc,g')$  & Polymatroids determined by subset $\Uulc\subseteq\Eulc$ and rank functions $f(.),f'(.),g(.),g'(.)$, respectively \\  \hline $h(.)$ &  Set function defined such that $h(\Uulc),\forall\;\Uulc\subseteq\Eulc$ yields the maximum weighted sum rate over polymatroid $\Pulk(\Uulc,f')$ & $h'(.)$ &  Set function defined such that $h'(\Uulc),\forall\;\Uulc\subseteq\Eulc$ yields the maximum weighted sum rate over polymatroid $\Tulk(\Uulc,g')$\\ \hline 
     \end{tabular}
\end{center}
\end{footnotesize}
\end{table}

%

\begin{figure}[htb!]
\begin{center}
\centerline{\epsfig{file=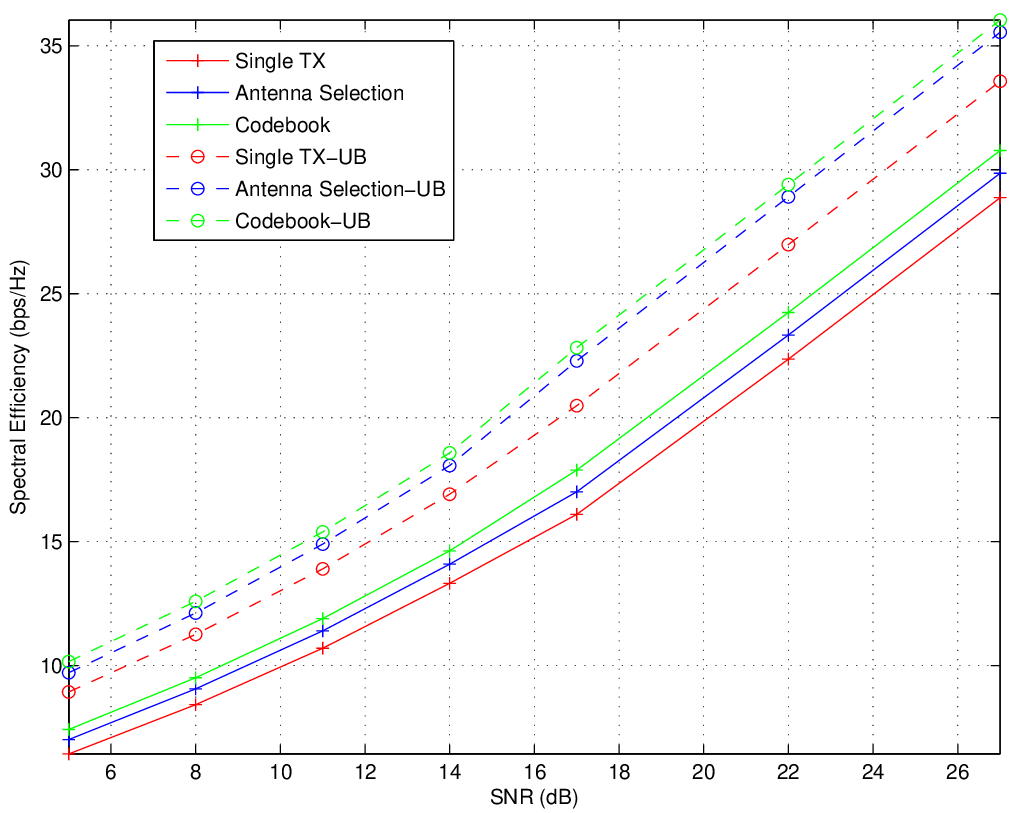,width=.6\linewidth}}
\caption{Average spectral efficiency versus SNR (dB).}
\label{fig_plot2}
\end{center}
\end{figure}

\begin{figure}[htb!]
\begin{center}
\centerline{\epsfig{file=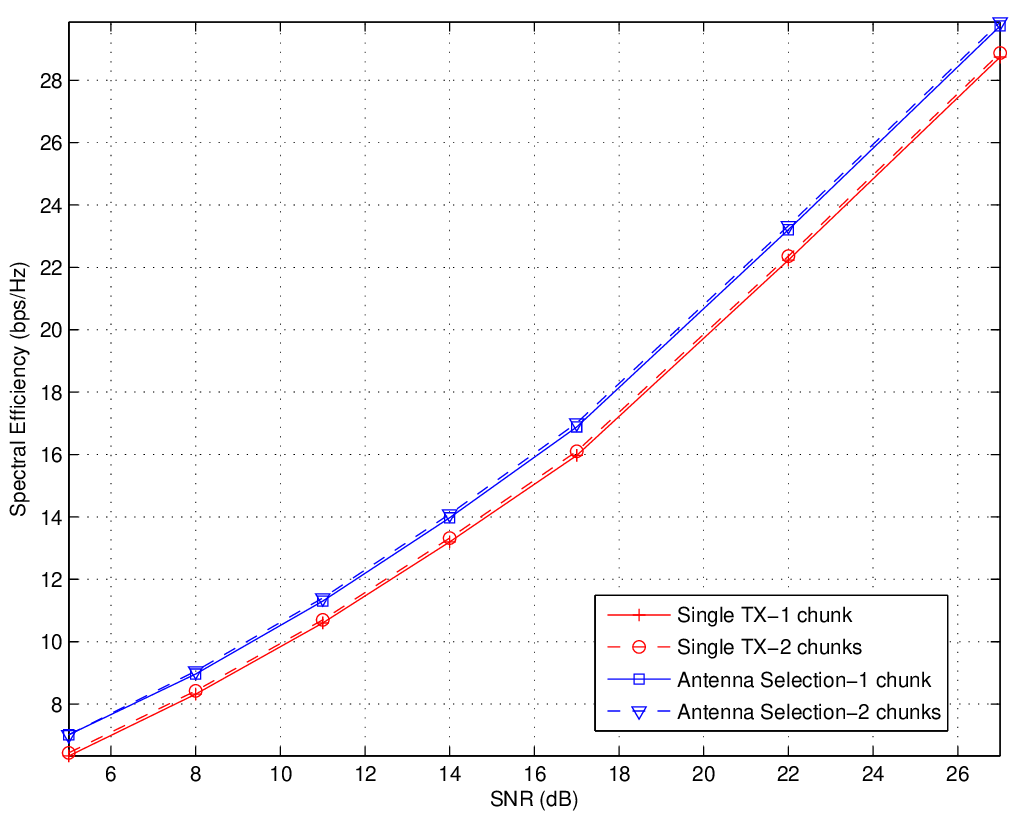,width=.6\linewidth}}
\caption{Impact of the number of chunks per user.}
\label{fig_plot4}
\end{center}
\end{figure}

\begin{figure}[htb!]
\begin{center}
\centerline{\epsfig{file=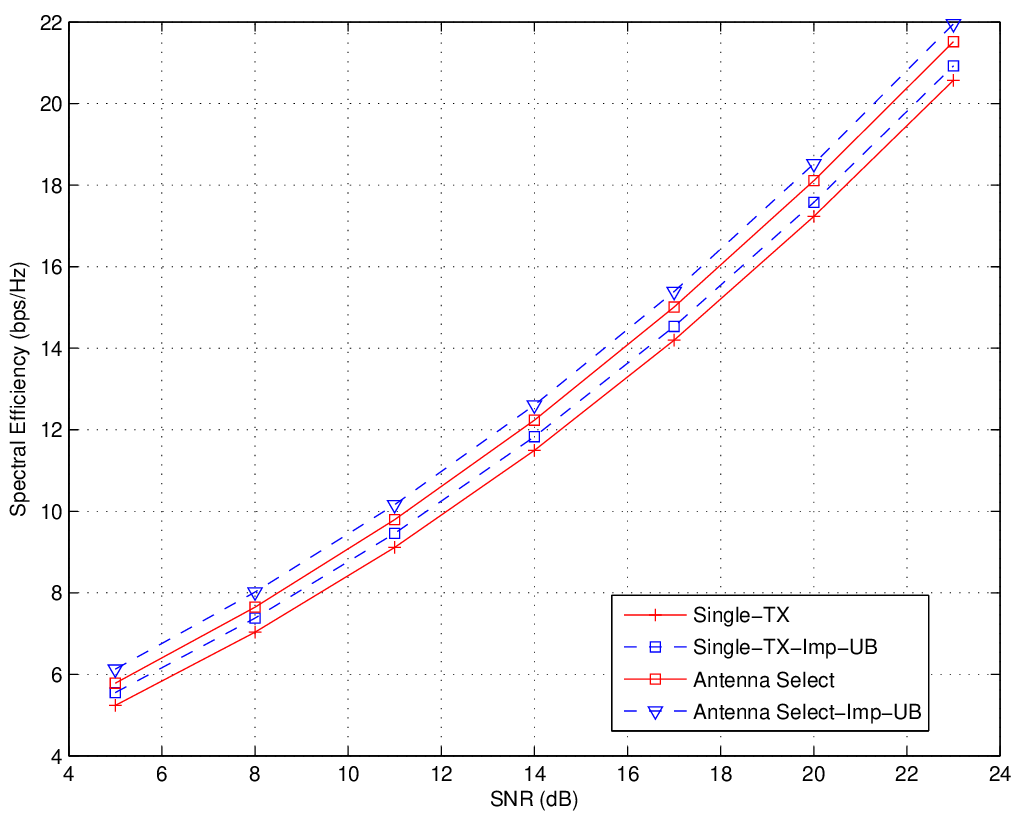,width=.6\linewidth}}
\caption{Convex Optimization based Upper Bound}
\label{fig_newUB}
\end{center}
\end{figure}

\begin{figure}[htb!]
\begin{center}
\centerline{\epsfig{file=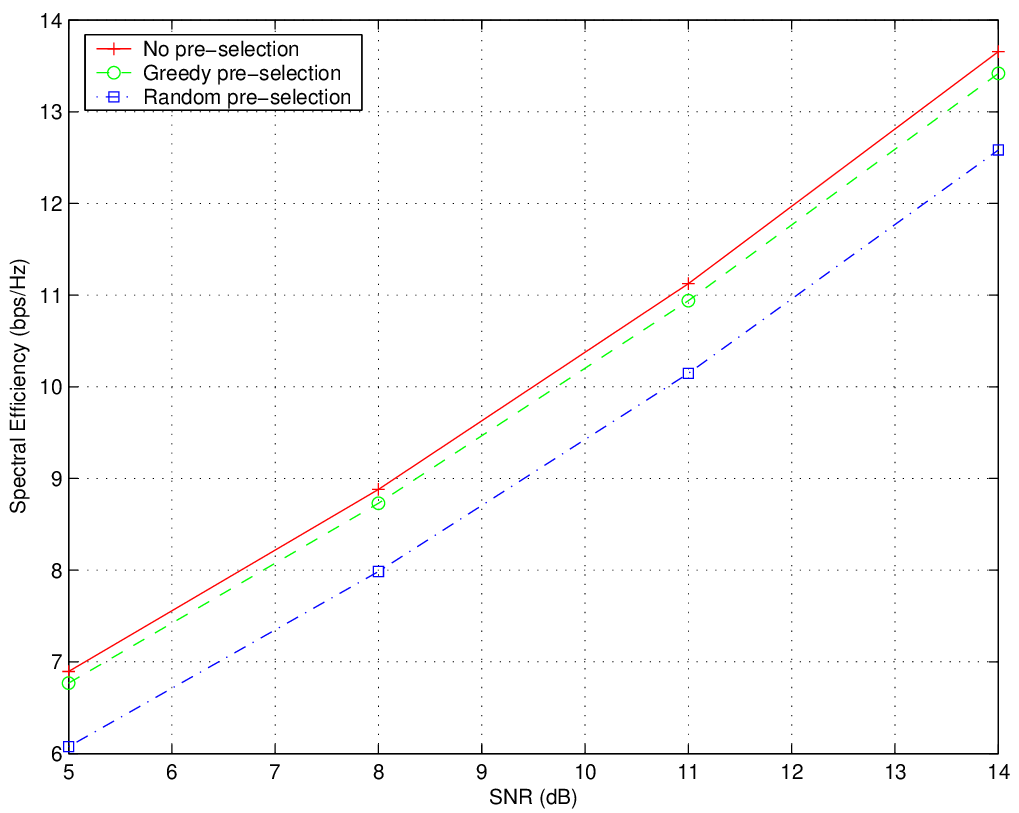,width=.6\linewidth}}
\caption{Impact of User Pre-Selection.}
\label{fig_plot3}
\end{center}
\end{figure}

\begin{figure}[htb!]
\begin{center}
\centerline{\epsfig{file=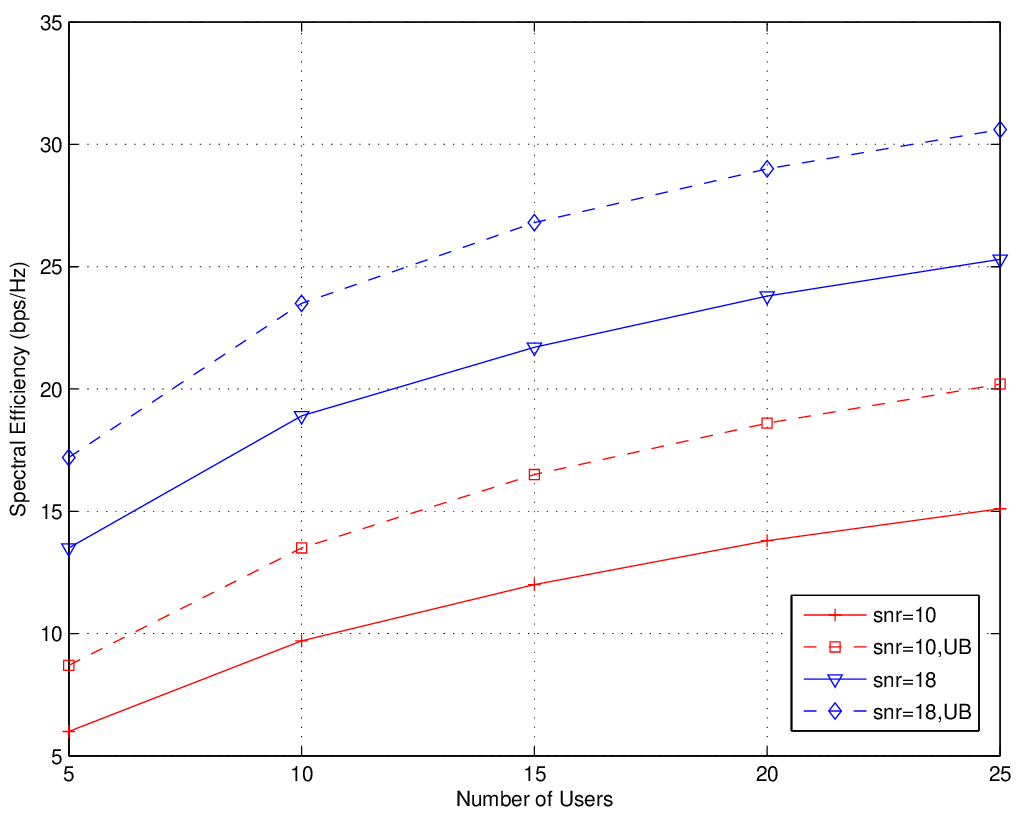,width=.6\linewidth}}
\centerline{\small (a)}
\centerline{\epsfig{file=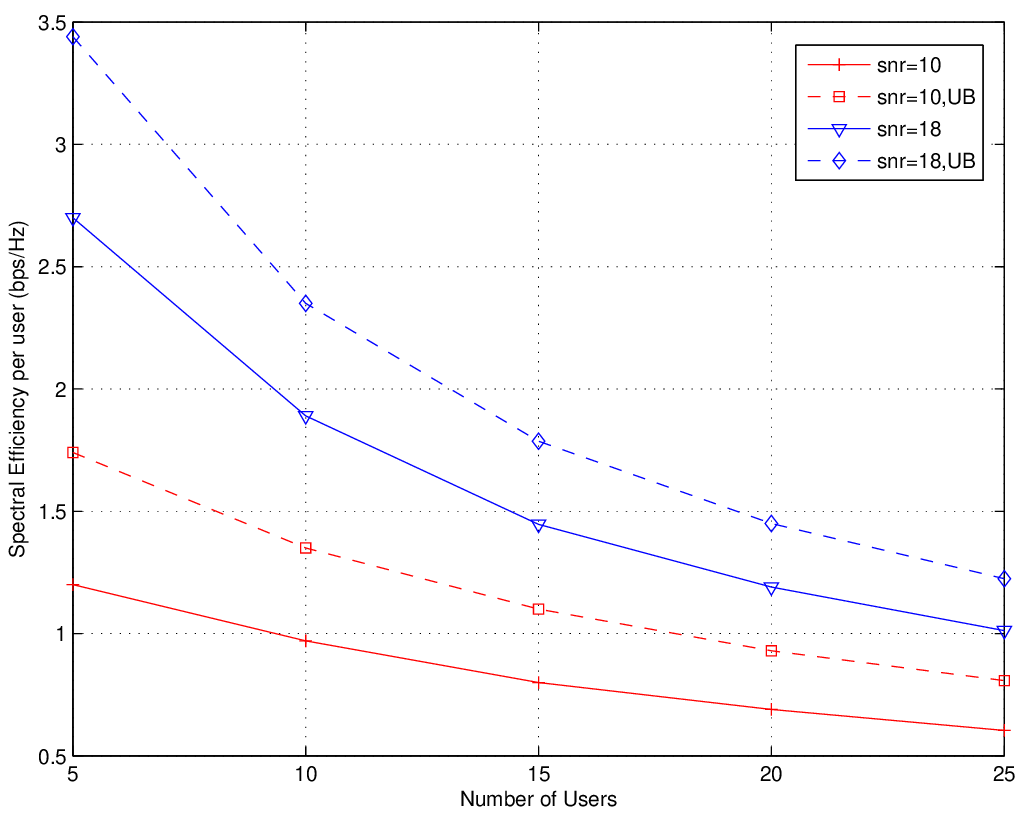,width=.6\linewidth}}
\centerline{\small (b)}
\caption{Impact of the number of users: (a) Cell spectral efficiency (b) per-user spectral efficiency.}
\label{fig_plotincu}
\end{center}
\end{figure}

\begin{figure}[htb!]
\begin{center}
\centerline{\epsfig{file=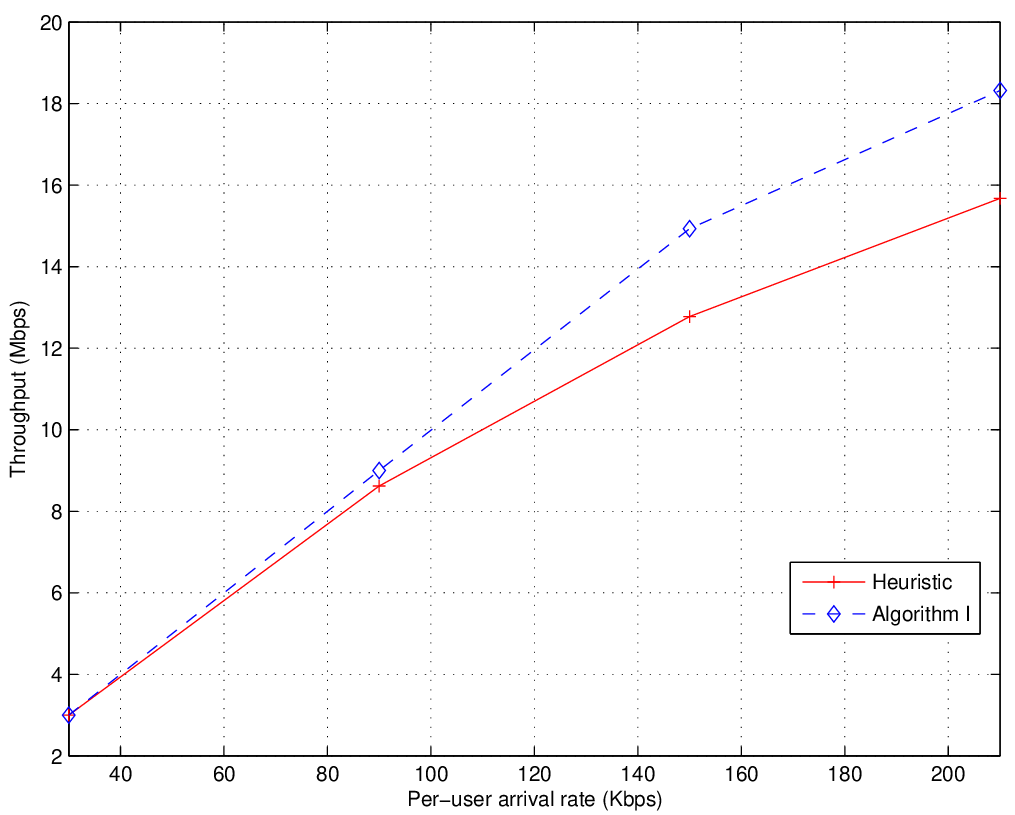,width=.6\linewidth}}
\caption{Impact of finite buffers.}
\label{fig_plotfb}
\end{center}
\end{figure}

\begin{figure}[htb!]
\begin{center}
\centerline{\epsfig{file=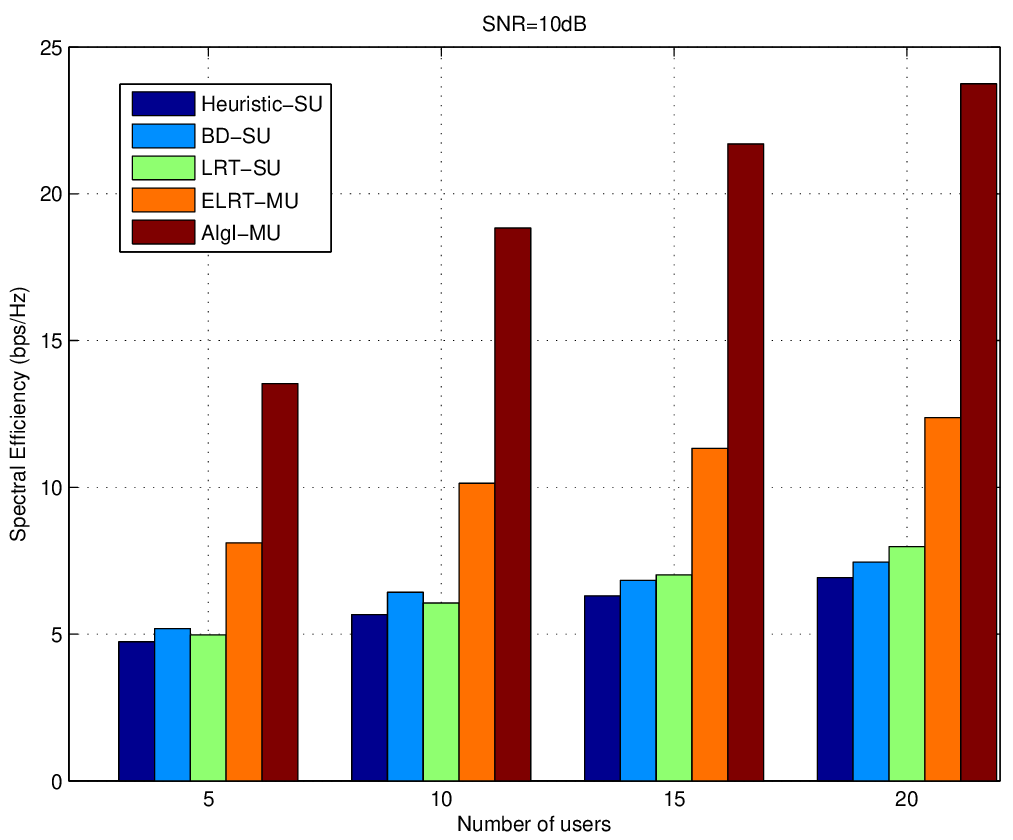,width=.6\linewidth}}
\centerline{\small (a)}
\centerline{\epsfig{file=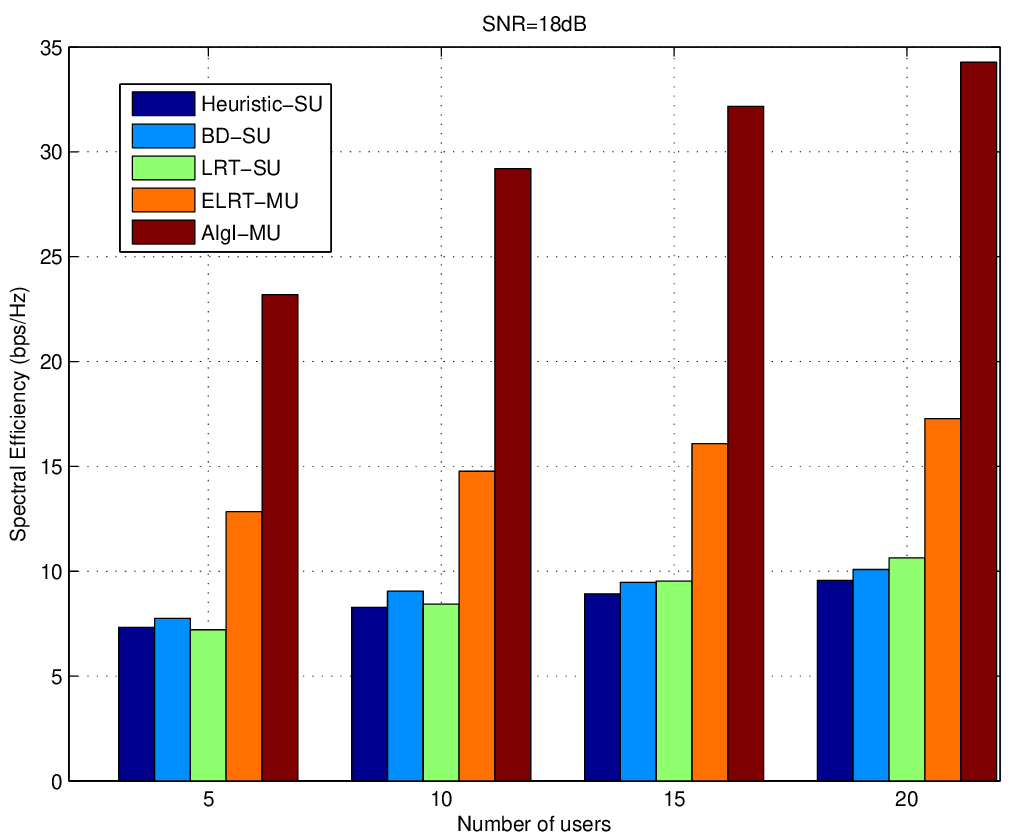,width=.6\linewidth}}
\centerline{\small (b)}
\caption{Comparison for varying number of users: (a) Average SNR=10dB (b) Average SNR=18dB .}
\label{fig_plotcompnew}
\end{center}
\end{figure}

%

%
%
%
%

\end{document}